\documentclass[sigconf]{acmart}
\pdfoutput=1

\usepackage{xcolor}
\usepackage{graphicx}
\usepackage[inline]{enumitem} 
\usepackage[font={small,bf},labelfont=bf,format=plain]{caption}
\usepackage{textcomp}
\usepackage{tabularx}
\usepackage{environ}
\usepackage{amsmath,amssymb,amsthm}
\usepackage{textgreek}
\usepackage{booktabs}
\usepackage{subfig}
\usepackage{xspace}
\usepackage[export]{adjustbox}
\usepackage{wrapfig}
\usepackage{amsthm}
\usepackage{mathrsfs}
\usepackage{mathtools}
\usepackage{algorithm}
\usepackage{algpseudocode}
\algtext*{EndProcedure}
\algtext*{EndFor}
\algtext*{EndWhile}
\algtext*{EndIf}
\usepackage{colortbl}
\usepackage{footnote}
\usepackage{url}
\usepackage{tabularx}
\usepackage{multicol}
\usepackage{multirow}
\usepackage{pifont}
\usepackage{bm}
\usepackage{listings}
\usepackage{wrapfig}
\usepackage{relsize}
\usepackage{mathtools}
\usepackage{bbm}
\usepackage{float}
\usepackage{blindtext}
\usepackage{balance}
\usepackage{outlines}

\usepackage{dsfont}
\usepackage{xfrac}

\newcolumntype{H}{>{\setbox0=\hbox\bgroup}c<{\egroup}@{}}

\makeatletter
\def\smallunderbrace#1{\mathop{\vtop{\m@th\ialign{##\crcr
   $\hfil\displaystyle{#1}\hfil$\crcr
   \noalign{\kern3\p@\nointerlineskip}
   \tiny\upbracefill\crcr\noalign{\kern3\p@}}}}\limits}
\makeatother

\newcolumntype{C}[1]{>{\centering\let\newline\\\arraybackslash\hspace{0pt}}m{#1}}

\newtheorem{problem}{Problem}

\newtheorem{definition}{Definition}
\newtheorem{lemma}{Lemma}
\newtheorem{theorem}{Theorem}

\newcommand{\ul}{\underline}
\newcommand{\xaxis}{$x$-axis\xspace}
\newcommand{\yaxis}{$y$-axis\xspace}
\newcommand{\pf}{\textsc{PvF}\xspace}

\newcommand{\method}{PENminer\xspace}
\newcommand{\methodOffline}{oPENminer\xspace}
\newcommand{\methodOnline}{sPENminer\xspace}

\newcommand{\graph}{G}
\newcommand{\nodes}{V}

\newcommand{\relation}{r}

\newcommand{\numNodes}{|\nodes|}
\newcommand{\edges}{E}
\newcommand{\numEdges}{|\edges|}

\newcommand{\evoGraph}{\mathcal{G}}

\newcommand{\update}{\ul{u}}
\newcommand{\newUpdate}{\update_{\text{new}}}
\newcommand{\timestamp}{t}
\newcommand{\updateType}{\pm}
\newcommand{\updateTimestamp}{\textsc{timestamp}}
\newcommand{\stream}{\mathscr{S}}
\newcommand{\streamStart}{\textsc{start}(\stream)}
\newcommand{\streamEnd}{\textsc{end}(\stream)}
\newcommand{\streamLen}{|\stream|}

\newcommand{\pattern}{x}
\newcommand{\patternSize}{k}
\newcommand{\extractedPatterns}{\mathcal{X}_{\text{extracted}}}

\newcommand{\duration}{\delta}

\newcommand{\view}{\phi}
\newcommand{\idView}{\textsc{ID}\xspace}
\newcommand{\labelView}{\textsc{Label}\xspace}
\newcommand{\orderView}{\textsc{Order}\xspace}

\newcommand{\window}{\mathscr{W}}
\newcommand{\windowWidth}{w}
\newcommand{\streamRate}{\mu}
\newcommand{\maxDuration}{\duration_{\max}}
\newcommand{\maxSize}{\patternSize_{\max}}

\newcommand{\order}{O}

\newcommand{\x}{x}
\newcommand{\uniformx}{\tilde{\x}}
\newcommand{\xUniverse}{\mathcal{X}}
\newcommand{\intervalStart}{t_s}
\newcommand{\intervalEnd}{t_e}
\newcommand{\interval}{[\intervalStart, \intervalEnd]}
\newcommand{\intervalUniverse}{\mathcal{I}}
\newcommand{\intervalWidth}{|\interval|}
\newcommand{\occ}{t}
\newcommand{\occs}{\mathcal{O}}
\newcommand{\uniqueOccs}{\tilde{\occs}}
\newcommand{\firstOcc}{t_f}
\newcommand{\lastOcc}{t_l}
\newcommand{\occsLong}{\{\firstOcc,\dots,\lastOcc\}}

\newcommand{\occInterval}{[\firstOcc,\lastOcc]}
\newcommand{\occIntervalsmall}{[t_s',t_e']}
\newcommand{\occIntervalWidth}{|\occInterval_\x|}

\newcommand{\freq}{F}
\newcommand{\freqLong}{\freq\big(x;\interval\big)}
\newcommand{\width}{W}
\newcommand{\widthLong}{\width\big(\pattern; \interval\big)}

\newcommand{\spread}{S}
\newcommand{\spreadLong}{\spread\big(\pattern; \interval\big)}
\newcommand{\entropy}{H}

\newcommand{\persistence}{P}
\newcommand{\persistenceLong}{P\big(x;[\intervalStart,\intervalEnd]\big)}
\newcommand{\gaps}{\Gamma}

\newcommand{\gap}{g}

\newcommand{\wExp}{\alpha}
\newcommand{\fExp}{\beta}
\newcommand{\sExp}{\gamma}

\newcommand{\axiomOne}{\textbf{A1}}
\newcommand{\axiomTwo}{\textbf{A2}}
\newcommand{\axiomThree}{\textbf{A3}}
\newcommand{\axiomFour}{\textbf{A4}}

\newcommand{\propOne}{\textbf{P1}}
\newcommand{\propTwo}{\textbf{P2}}
\newcommand{\propThree}{\textbf{P3}}

\newcommand{\rOne}{\textbf{RQ1}}
\newcommand{\rTwo}{\textbf{RQ2}}
\newcommand{\rThree}{\textbf{RQ3}}

\newcommand{\sedanspot}{\textsc{SedanSpot}\xspace}
\newcommand{\midas}{\textsc{Midas-R}\xspace}
\newcommand{\ds}{\textsc{DS}\xspace}
\newcommand{\freqBaseline}{\textsc{Freq}\xspace}

\newcommand{\timestampUniverse}{\mathbb{R}_{\geq 0}}

\newcommand{\euemail}{Eu Email\xspace}
\newcommand{\columbus}{Columbus Bike\xspace}
\newcommand{\chicago}{Chicago Bike\xspace}
\newcommand{\boston}{Boston Bike\xspace}
\newcommand{\stackoverflow}{Stackoverflow\xspace}
\newcommand{\darpa}{DARPA IP\xspace}
\newcommand{\NYC}{NYC Taxi\xspace}
\newcommand{\reddit}{Reddit\xspace}

\newcommand{\odds}{\texttt{\small r/nightly\_pick}\xspace}
\newcommand{\hockey}{\texttt{\small r/hockey}\xspace}
\newcommand{\cig}{\texttt{\small r/electronic\_cigarette}\xspace}
\newcommand{\poker}{\texttt{\small r/ecrpoker}\xspace}
\newcommand{\bestof}{\texttt{\small r/bestof}\xspace}
\newcommand{\finance}{\texttt{\small r/personalfinance}\xspace}

\newcommand{\taxiCommision}{\texttt{\small zone207}\xspace}
\newcommand{\suspiciousSrc}{\texttt{\small zone135}\xspace}
\newcommand{\suspiciousDst}{\texttt{\small zone170}\xspace}
\newcommand{\nothingSrc}{\texttt{\small zone234}\xspace}
\newcommand{\nothingDst}{\texttt{\small zone198}\xspace}

\copyrightyear{2020}
\acmYear{2020}
\setcopyright{acmcopyright}\acmConference[KDD '20]{Proceedings of the 26th ACM SIGKDD Conference on Knowledge Discovery and Data Mining}{August 23--27, 2020}{Virtual Event, CA, USA}
\acmBooktitle{Proceedings of the 26th ACM SIGKDD Conference on Knowledge Discovery and Data Mining (KDD '20), August 23--27, 2020, Virtual Event, CA, USA}
\acmPrice{15.00}
\acmDOI{10.1145/3394486.3403136}
\acmISBN{978-1-4503-7998-4/20/08}

\begin{document}
\fancyhead{}
\title{Mining Persistent Activity in Continually Evolving Networks}

\author{Caleb Belth}
\affiliation{
  \institution{University of Michigan}
}
\email{cbelth@umich.edu}

\author{Xinyi Zheng}
\affiliation{
  \institution{University of Michigan}
}
\email{zxycarol@umich.edu}

\author{Danai Koutra}
\affiliation{
  \institution{University of Michigan}
}
\email{dkoutra@umich.edu}

\begin{abstract}
Frequent pattern mining is a key area of study that gives insights into the structure and dynamics of evolving networks, such as social or road networks. 
However, not only does a network evolve, but often the \emph{way} that it evolves, itself evolves. Thus, knowing, in addition to patterns' frequencies, for how long and how regularly they have occurred---i.e., their \emph{persistence}---can add to our understanding of evolving networks. 
In this work, we propose the problem of mining activity that persists through time in continually evolving networks---i.e., activity that repeatedly and consistently occurs. 
We extend the notion of temporal motifs to capture activity among \emph{specific} nodes, in what we call \emph{activity snippets}, which are small sequences of edge-updates 
that reoccur. We propose axioms and properties that a measure of persistence should satisfy, and develop such a persistence measure. We also propose \method, an efficient framework for mining activity snippets' \emph{Persistence in Evolving Networks}, and 
design both offline and streaming algorithms. We apply \method to numerous real, large-scale evolving networks and edge streams, and 
find activity that is surprisingly regular over a long period of time, 
but too infrequent to be discovered by aggregate count alone, and bursts of activity exposed by their lack of persistence. Our findings with \method include neighborhoods in NYC where taxi traffic persisted through Hurricane Sandy, 
the opening of new bike-stations, characteristics of social network users, and more. 
Moreover, we use \method towards identifying anomalies in multiple networks, outperforming baselines at identifying subtle anomalies by 9.8-48\% in AUC.
\end{abstract}

\maketitle

\section{Introduction}
\label{sec:intro}
Many networks evolve continually over time, such as traffic networks that encode current en-route traffic, networks of computers (IP-addresses) sending messages to each other, social networks of users interacting over time, and more. 
In order to understand these networks and the systems they represent, it is important to consider not just their structure, but also their temporal dynamics. Towards this goal, existing works have focused on counting 
patterns in networks. These patterns include temporal motifs (small subgraph patterns) \cite{kovanen2011temporal, gurukar2015commit, paranjape2017motifs, Liu2019SamplingMF}, frequent subgraphs in evolving networks \cite{aggarwal2010dense, abdelhamid2017incremental, aslay2018mining}, flow motifs \cite{Kosyfaki2018FlowMI}, communication motifs \cite{Zhao2010CommunicationMA}, and coevolving relational motifs \cite{ahmed2015algorithms}. 
This line of research uses the frequency 
of patterns towards understanding the behavior of networks (e.g., some interaction patterns are more common in Q\&A forums than in instant messengers \cite{paranjape2017motifs}). However, many patterns in evolving networks may only last for a short period of time (e.g., bursty activity). This can lead to large aggregate counts, making a bursty anomaly falsely appear to be an intrinsic characteristic of a dynamic network.
On the other hand, there may be important activity in a network that has low overall frequency, but that occurs continually and regularly~\cite{timecrunch}, such as a stealthy computer-network attack. Thus, to more fully understand the dynamics of evolving networks, a pattern's \emph{persistence}---involving how long it has occurred, how uniformly its occurrences are spread out, \emph{and} what its frequency is---should be taken into account (Fig.~\ref{fig:main}).

\begin{figure}[t]
    \centering
    \includegraphics[width=0.9\columnwidth]{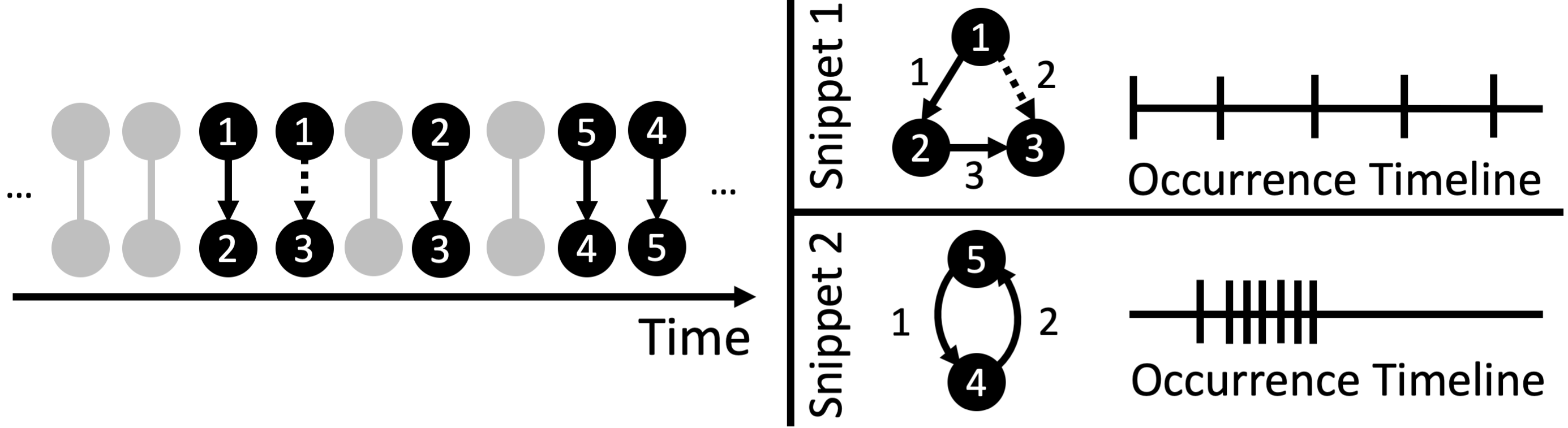}
    \caption{We seek to measure the persistence of activity snippets (i.e., reoccurring sequences of edge-updates) in continually evolving networks, which allows snippets 1 and 2 to be differentiated. While both snippets have approximately the same frequency, snippet 1 is more persistent, and snippet 2 is bursty. Activity snippets differ from temporal motifs by allowing different nodes interacting in the same way to be treated as distinct snippets 
    {(e.g., depending on the application, if snippet 2 reoccurs with nodes 6 and 7, it can either be treated as an instance of the same snippet, or of a new snippet).}
    }
    \label{fig:main}
    \vspace{-0.5cm}
\end{figure}

In this work, we seek to mine persistent activity in continually evolving networks. We view the activity in a network as a stream of edge \emph{updates} (insertions or deletions) over time. We introduce activity snippets by extending the idea of temporal motifs \cite{paranjape2017motifs}, which encode activity among nodes in \emph{general}, to allow for encoding activity between or among \emph{specific} nodes. This extension allows us to capture exact activity, (e.g., a specific edge in the network), which is important for applications like identifying which public transportation routes are used persistently, or identifying suspicious nodes or edges (e.g., nodes 4 and 5 in Fig.~\ref{fig:main}). In other applications, such as social network analysis, we may be interested in interactions between nodes in general, with no specific users in mind. In this case, activity snippets reduce to temporal motifs. In either case, activity can \emph{repeat}, whether it be specific nodes repeatedly engaging,
or different sets of nodes behaving in consistent ways. We seek to analyze this activity in terms of how persistently it occurs.

Our main contributions are as follows:

\begin{itemize}[leftmargin=*]
    \item \textbf{Precise Formulation of Persistence:} We introduce four axioms and three general properties that a measure of persistence ought to satisfy. We also provide a versatile persistence measure, and prove that it satisfies all axioms and properties.
    \item \textbf{Offline and Online Algorithms:} We develop \method, an efficient framework that uses our persistence measure to mine evolving networks. Our offline version, \methodOffline, uses the measure for analyzing time-stamped sequences of edges from the past. Our online version, \methodOnline, computes the measure incrementally for real-time analysis of edge-streams.  
    \item \textbf{Extensive Experiments on Real Data:} We perform experiments on real, large scale evolving networks, which reveal real-world phenomena, including neighborhoods in NYC where taxi traffic persisted through Hurricane Sandy, the opening of new bike-stations in multiple cities, characteristics of social network users, and more. We also demonstrate that \method is scalable in offline and streaming settings, 
    {processing edge-updates in each stream 10K to 360K times faster than the rate of that stream}. 
    This allows both subtle and bursty activity to be identified in real time, right when it happens. 
    \method also effectively identifies subtle anomalies, outperforming baselines by 9.8-48\% in AUC.
\end{itemize}

{Our code is available at {\href{https://github.com/GemsLab/PENminer}{https://github.com/GemsLab/PENminer}}}.

\section{Related Work}
\label{sec:related}

Our work is related to motif mining, frequent subgraph mining, and persistent item mining. 

\vspace{0.1cm}
\noindent
\textbf{Motif Mining.}
\label{subsec:motif}
Network motifs, which are small, frequent subgraph patterns~\cite{Milo2002NetworkMS} 
are used in various problems, such as network summarization~\cite{LiuSDK18-survey} and exploration, and dense subgraph detection \cite{specgreedy}. 
Temporal motifs~\cite{kovanen2011temporal} extend static motifs to evolving networks by adding an ordering to their edges. 
Researchers have analyzed temporal motifs within and across a wide-range of networks \cite{paranjape2017motifs}, and have proposed 
sampling methods to estimate their counts \cite{Liu2019SamplingMF}. Temporal motifs have also been extended to capture the flow of information among nodes in a motif~\cite{Kosyfaki2018FlowMI}. Among different motifs, triangle counting has
attracted significant interest \cite{stefani2017triest}.

A special case, communication motifs, was introduced to model temporal communication patterns of users in communication networks, for the purpose of understanding information flow~\cite{Zhao2010CommunicationMA,gurukar2015commit}. In~\cite{Zhao2010CommunicationMA}, focusing on motif counts, the authors observed stability in the ranking of the ten most frequent motifs over time, but did not investigate subtly persistent or bursty behavior. 
Coevolving Relational Motifs (CRMs) are patterns that describe sets of nodes that evolve together in a consistent way \cite{ahmed2015algorithms}. In CRMs, consistency means that when a set of nodes co-evolve, it is almost always in the way specified by the motif. Thus, it considers the \emph{relative} frequency of the motif, compared to other possible motifs over the same nodes, but does not investigate persistence. 

Our work differs from motif-mining in its focus on \emph{persistence} not  just counts, and much of our analysis deals with \emph{specific} edges or sequences of edges, rather than general activity among nodes.

\vspace{0.1cm}
\noindent
\textbf{Frequent Subgraph Mining}. Frequent subgraph mining (FSM) has two settings: transactional FSM and single graph FSM \cite{jiang2013survey}. The transactional setting attempts to identify, in a sequence of graphs (transactions), subgraphs that appear in a large number of the graphs. 
This setting naturally extends to evolving networks, by treating graph snapshots 
(corresponding to a batch of edge updates) as transactions  ~\cite{ray2014frequent}.
The single graph setting seeks to identify subgraphs that have many instances within a single, large graph \cite{elseidy2014grami}. 
It has also been extended to evolving networks, including edge streams \cite{aggarwal2010dense, abdelhamid2017incremental, aslay2018mining}, where the goal is to adaptively \emph{maintain} the most frequent subgraphs. 
While these methods seek to find \emph{frequent subgraphs}, we seek to find \emph{persistent activity}.

\vspace{0.1cm}
\noindent
\textbf{Persistent Item Mining.} Motivated by scenarios of \textit{stealthy} click-fraud or distributed denial-of-service attacks, persistent item mining in data streams was introduced in \cite{dai2016finding}. 
Arguing that, besides frequency, the length of the time period in which an item appears is important for understanding the dynamics of streaming data, the authors introduced a heuristic definition of a persistent item as one that occurs at least once in a large number of equally-sized observation periods.
As we show, this simple definition violates some intuitive desired properties (\S~\ref{subsec:measure-theory}). Our work goes beyond heuristics to establish a technical definition of a persistence measure, and focuses on edge-streams. 

Persistent community detection, which attempts to find communities that last for long durations of time, has also been studied~\cite{semertzidis2019finding,li2018persistent}. However, these works seek to find tightly-knit subgraphs (communities) that last for a long time, while we focus on measuring \textit{activity} that regularly \textit{re-occurs} through time 
without requiring that this activity take place in contiguous chunks of time.

\section{Theory}
\label{sec:theory}
\subsection{Preliminary Definitions}
\label{subsec:prelim}
 We begin with preliminary definitions. 
Since most definitions apply beyond network settings, 
we first introduce the general concepts in \S~\ref{subsubsec:prelim-general} and then the network-specific terminology in \S~\ref{subsubsec:prelim-network}.

\begin{table}[!t]
\caption{Description of major symbols.}
\label{tab:Symbols}
\vspace{-0.2cm}
\centering
\resizebox{\columnwidth}{!}{
  \begin{tabular}{ll}
  \toprule
     \textbf{Notation} & \textbf{Description}  \\ \midrule
     $\stream$ & Edge stream \\
     $\pattern$ & Activity snippet\\
     $\maxSize$, $\maxDuration$ & Maximum size and duration for a snippet\\
     $\view$ & View of a snippet (\idView, \labelView, or \orderView)\\
     $[t_i,t_j]$, $|[t_i,t_j]|$ & Interval of time and its width $t_j - t_i$\\
     $\occs_\x$, $\uniqueOccs_\x$ & All and unique occurrences of $\x$ in $\interval$, resp.\\
     $\gaps_\x$ & Gaps (lengths of time) between unique occurrences of $\x$\\
     \bottomrule
    \end{tabular}
    \vspace{-0.7cm}
    }
\end{table}

\subsubsection{Events in Time}
\label{subsubsec:prelim-general}

\vspace{0.1cm}
\noindent 
\textbf{Interval of Time.} An interval of time $\interval$ (Fig.~\ref{fig:intervals}) is defined as $\{\timestamp : \timestamp \in \mathbb{R}_{\geq 0}, \timestamp \in \interval\}$.
The width of the interval is $|\interval| = \intervalEnd - \intervalStart \geq 0$. The set of all intervals is $\intervalUniverse \triangleq \{[\intervalStart, \intervalEnd] : \intervalStart, \intervalEnd \in \timestampUniverse, \intervalStart \leq \intervalEnd\}$.

\begin{wrapfigure}[3]{r}{0.3\linewidth}
    \vspace{-0.9cm}
    \centering
    \includegraphics[width=0.7\linewidth]{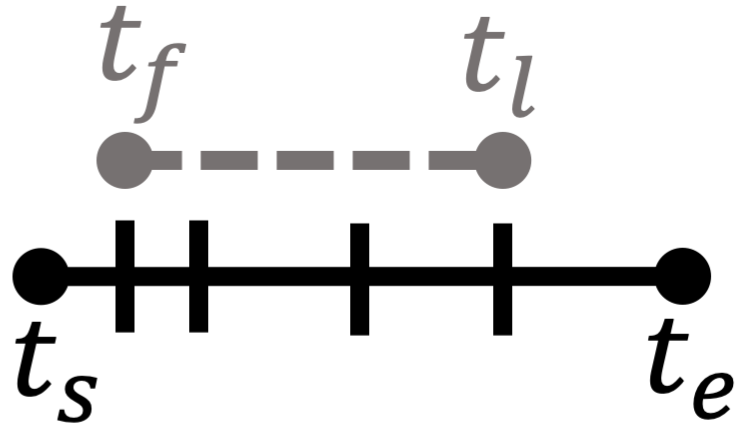}
    \vspace{-0.35cm}
    \caption{Intervals.} 
    \label{fig:intervals}
    \vspace{-0.4cm}
\end{wrapfigure}

\vspace{0.1cm}
\noindent 
\textbf{Event.} An event $\x$, with respect to an interval $\interval$, is something that occurs at least once
in that interval. Examples include item purchases in a sequence of transactions, 
and measurements in a time series. The universe of events is $\xUniverse$.

\vspace{0.1cm}
\noindent 
\textbf{Event Occurrence}. 
An occurrence, $\occ \in \interval$ of $\x$, is a timestamp. 
All the occurrences of $\x$, $\occs_\x = \occsLong$, form an ordered multiset of timestamps between the first and last occurrrences, and the corresponding \emph{interval of occurrences} is $\occInterval_\x \subseteq \interval$ (e.g., Fig.~\ref{fig:intervals}). 
We denote the ordered set of \emph{unique} occurrences $\uniqueOccs_\x$.

\vspace{0.1cm}
\noindent 
\textbf{Occurrence Gaps}. The gaps (i.e., amount of time) between occurrences of $\x$ 
form the sequence $\gaps_x = (\gap_1, \dots, \gap_{|\uniqueOccs_\x| - 1})$, where $\gap_i = t_{i+1} - t_{i}$ for $t_i, t_{i+1} \in \uniqueOccs_\x$. The number of gaps is $|\gaps_\x| = |\uniqueOccs_\x| - 1$.

\subsubsection{Activity Snippets in Evolving Networks}
\label{subsubsec:prelim-network}

\textbf{Graph or Network.} A graph or network $\graph = (\nodes, \edges)$ is a set of nodes $\nodes$, and a set of edges $\edges \subseteq \nodes \times \nodes$. 
If an edge between nodes $v_1$ and $v_2$ has a relationship type $r$, we denote it $(v_1, r, v_2)$. 
Our proposed persistence measure and algorithms apply to general networks: labeled (nodes/edges), directed, weighted, bipartite or multi-graphs.

\vspace{0.1cm}
\noindent 
\textbf{Edge-update.} An edge update $\update = (\updateType, v_1, \relation, v_2, \timestamp)$ to a network $\graph$ is the \emph{insertion} (+) or \emph{deletion} (-) of an edge $(v_1,r,v_2)$ at time $\timestamp$. 
We refer to $\update$'s timestamp with $\updateTimestamp(\update)$.

\vspace{0.1cm}
\noindent 
\textbf{Edge Stream.} An edge stream $\stream = (\update_1, \update_2, \dots)$ is a time-ordered sequence of possibly infinite edge-updates to an evolving graph $\evoGraph = (\graph_1, \graph_2, \dots)$. We denote the start time of $\stream$, $\streamStart = \updateTimestamp(\update_1)$, its length $|\stream|$, and if bounded, its end time $\streamEnd = \updateTimestamp(\update_{|\stream|})$. We call a sub-sequence of $\stream$ that consists of 
all edge-updates in the last $\windowWidth$ time units a \textbf{window} $\window$ of width $\windowWidth$.

\begin{figure}[t]
    \centering
    \includegraphics[width=0.7\columnwidth]{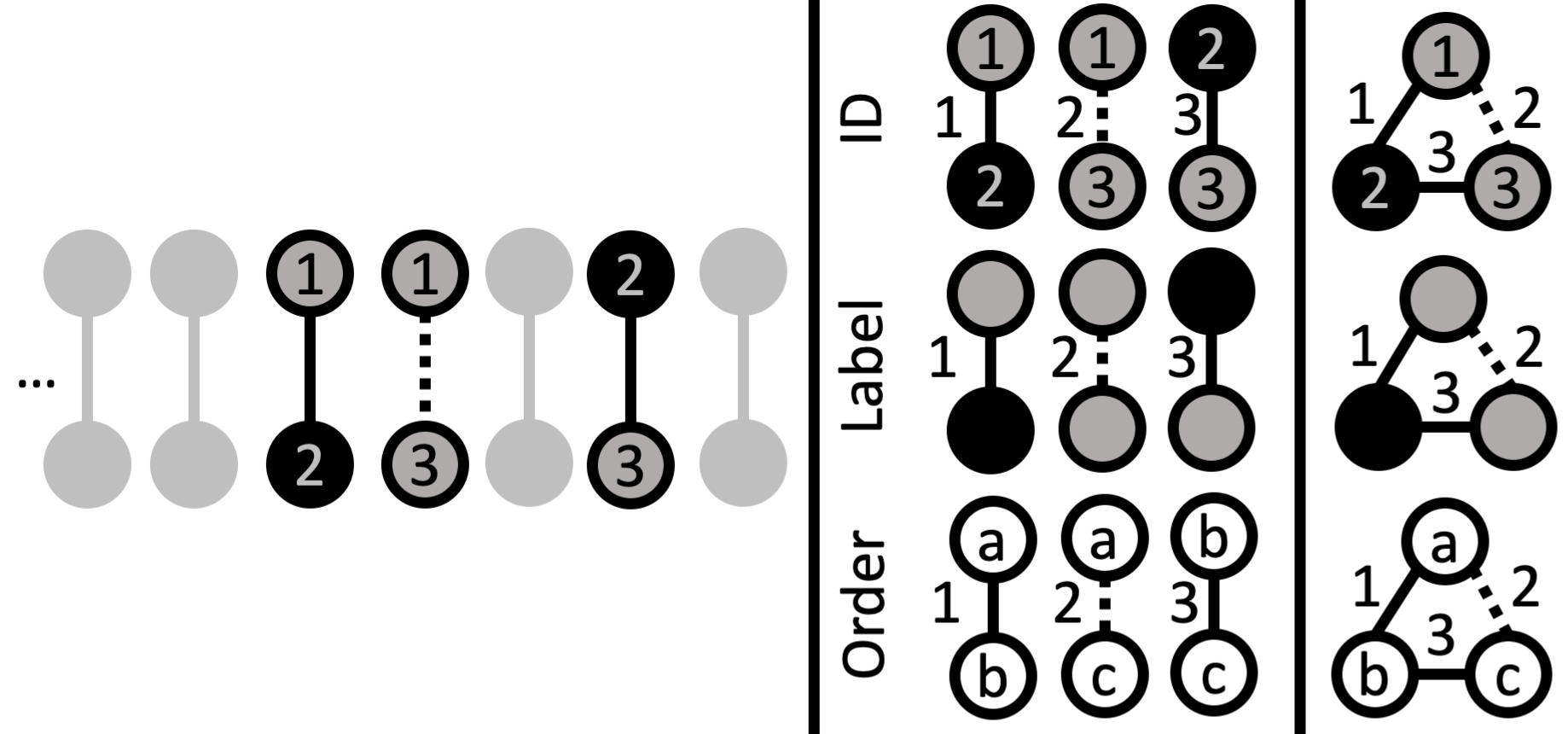}
    \vspace{-0.25cm}
    \caption{Left: An update sequence that is ordered but not contiguous. Right: The activity snippet for the sequence, depending on the view $\view$ (\idView, \labelView, or \orderView), and a graph describing the activity. Solid/dashed edges capture types of interactions (e.g., edge types or insertions vs. deletions), and the node colors denote labels.}
    \label{fig:seq-view}
    \vspace{-0.3cm}
\end{figure}

\vspace{0.1cm}
\noindent 
\textbf{Activity Snippet.} Intuitively, an activity snippet $x$ describes a sequence of activity among connected nodes in the network. Specifically, $\pattern = (\view(\update_{i}), \dots, \view(\update_{j}))$ is a sequence of ordered, but \textit{not}-necessarily contiguous, edge updates, where the node IDs  \emph{may} be replaced by a \textbf{view} $\view$: (1) their labels or (2) the position of their first occurrence in the sequence (Fig.~\ref{fig:seq-view}).
If the node IDs are \textit{not} replaced, the snippet is an exact sequence of activity. 
If the IDs are replaced by their position in the sequence, then the snippet is analogous to temporal motifs \cite{paranjape2017motifs}, but also captures edge deletions or insertions. 
The snippet has duration $\duration = \updateTimestamp(\update_{j}) - \updateTimestamp(\update_{i})$ and size $\patternSize$ equal to the number of edge-updates in the snippet. A ($\maxDuration, \maxSize$) activity snippet is an activity snippet with duration $\duration \leq \maxDuration$ and size $\patternSize \leq \maxSize$. In this work, events (\S~\ref{subsubsec:prelim-general}) are activity snippets.

\vspace{0.1cm}
\noindent \textbf{Problem Definition.} Given these definitions, we focus on the problem of persistent activity mining:
\begin{problem}[Persistent Activity Mining]\label{problem:prob1}
Given a network $\graph$ that continually evolves via an edge-stream $\stream$, measure the persistence of each activity snippet $\pattern$, i.e., for how long, how often, and how regularly it has occurred.
\end{problem}

\subsection{Properties of Persistence}
\label{subsec:pers-theory}
Based on our definitions, we present axioms and properties that a persistence measure should follow. 
In the remaining sections, we propose a principled persistence measure (\S~\ref{subsec:pers-measure}) and prove that it satisfies all the axioms and properties (\S~\ref{subsec:measure-theory}). 
Although \textbf{our theoretical definitions, properties, proposed measure and derivations are general and apply to any event $\x$ in a stream or time series}, in the context of Problem~\ref{problem:prob1}, $\x$ is an activity snippet.

\begin{definition}[Persistence Measure]
\label{def:persistence-measure}
A persistence measure $\persistenceLong : \xUniverse \times \intervalUniverse \rightarrow \mathbb{R}_{\geq 0}$
is a function that defines 
the persistence of an event $x \in \xUniverse$ in the interval $\interval \in \intervalUniverse$, and 
satisfies the following axioms:
\begin{itemize}[leftmargin=*]
    \item \axiomOne: It is non-negative, and $0$ iff there are no occurrences. Formally, $\persistenceLong \geq 0$, with equality iff $|\occs_\x| = 0$.
    \item \axiomTwo: As the interval becomes infinitely filled with unique occurrences, persistence tends to infinity.
    That is, $\lim_{|\uniqueOccs_\x| \rightarrow \infty} \persistenceLong = \infty$.
    \item \axiomThree: Shifting all occurrences in time does not affect persistence.
        Formally, $\persistence\big(\x';\interval\big) = \persistenceLong$, where $\x'$ is an event with occurrences $\occs_{\x'} = \{\occ + c : \occ \in \occs_x\}$, for some constant $c\in\mathbb{R}$ such that the shifted points remain in $\interval$, i.e, $\firstOcc + c \geq \intervalStart, \lastOcc + c \leq \intervalEnd$.
     \item \axiomFour: Shrinking the interval of measurement towards $\occInterval_\x$ leads to higher persistence. 
         Mathematically, $\persistenceLong \leq \persistence\big(x; \occIntervalsmall \big) \leq \persistence\big(x; \occInterval_\x\big)$, 
         for $t_s \leq t_s' \leq t_f$ and $t_e \geq t_e' \geq t_l$. The equality holds for $t_s = t_s' = t_f$ and $t_e = t_e' = t_l$.
         
\end{itemize}
\end{definition}

Besides these axioms, there are several properties that a good persistence measures ought to follow. For use in presenting these properties, let $\uniformx^n$ be a special class of event with $n$ occurrences, all of which are unique and uniformly spaced out over $\occInterval_{\uniformx^n}$, 
that is, $|\occs_{\uniformx^n}| = |\uniqueOccs_{\uniformx^n}| = n$ and $\gap_i = \gap_j = \frac{|\occInterval_{\uniformx^n}|}{|\gaps_{\uniformx^n}|}, \forall \gap_i,\gap_j \in \gaps_{\uniformx^n}$.

\begin{itemize}[leftmargin=*]
    \item \propOne: For two events with $n$ unique, uniformly-spaced occurrences, persistence is larger for the event with occurrences spread over a wider interval. 
    Formally, for any  ${\uniformx_1}^n$ and ${\uniformx_2}^n$ such that $|\occInterval_{{\uniformx_1}^n}| < |\occInterval_{{\uniformx_2}^n}|$, it should hold that $\persistence\big({\uniformx_1}^n;\interval\big) < \persistence\big({\uniformx_2}^n;\interval\big)$.
    \item \propTwo: For two events with unique, uniformly-spaced occurrences spread out over the same interval, persistence is larger for the event with more occurrences. 
    Mathematically, for any $\uniformx^n$ and $\uniformx^{n+k}$ such that $\occInterval_{\uniformx^n} = \occInterval_{\uniformx^{n+k}}$, the persistence measure should satisfy $\persistence\big(\uniformx^n;\interval\big) < \persistence\big(\uniformx^{n+k};\interval\big), \forall k\geq1$.
    \item \propThree: The persistence of an event with $n$ unique occurrences in an interval is maximized iff the occurrences are spread out uniformly. 
    Formally, for any $\x$ and $\uniformx^n$ such that $|\occs_\x| = |\uniqueOccs_\x| = |\uniqueOccs_{\uniformx^n}| = n$ and $\occInterval_\x = \occInterval_{\uniformx^n}$, the persistence measure should satisfy $\persistenceLong \leq \persistence\big(\uniformx^n; \interval\big)$, with equality iff $\gaps_\x = \gaps_{\uniformx^n}$, $\forall n > 2$.
\end{itemize}

\subsection{Proposed Persistence Measure}
\label{subsec:pers-measure}

\noindent \textbf{Family of Persistence Measures.} The axioms and properties point to three main components of persistence: the width of the interval in which occurrences fall, the frequency, and the distribution of occurrences. We thus propose a family of persistence measures
\begin{equation}
    \label{eq:persistence-form}
    \persistence\big(\pattern;\interval\big) \triangleq f\Big( \underbrace{\widthLong}_{\text{width}}, \underbrace{\freqLong}_{\text{frequency}}, \underbrace{\spreadLong}_{\text{spread}} \Big),
\end{equation}
\noindent where $\width(\cdot)$ is a function of the width of $\x$'s interval of occurrences $\occInterval_\x$, 
$\freq(\cdot)$ is a function of the number of occurrences, 
$\spread(\cdot)$ is a function of how uniformly the points are spread out over $\occInterval_\x$, and
$f(\cdot)$ is a function that combines these components. 

\vspace{0.1cm}
\noindent \textbf{A Persistence Measure.} There are a number of ways to construct the functions $\width(\cdot)$, $\freq(\cdot)$, and $\spread(\cdot)$, and to combine them. We propose one intuitive persistence measure within the family described above and show that it satisfies all the axioms and properties (\S~\ref{subsec:measure-theory}).

First, we define $\width(\cdot)$ as the percentage of the interval width that is covered by the occurrences of $x$:
\begin{equation}
    \label{eq:width}
    \width\big(\pattern; \interval\big) \triangleq \frac{\occIntervalWidth + 1}{\intervalWidth + 1},
\end{equation}
where we add one in the numerator and denominator so that they are non-0 when $\firstOcc = \lastOcc$ or $\intervalStart = \intervalEnd$.
For $|\occs_\x|=0$, we define $W(\cdot)=0$.

Second, we define $\freq(\cdot)$, the function of frequency, as the logarithm of the number of occurrences $|\occs_\x|$ to: (a) prevent the term from dominating over the others (which can be more than an order of mangitude smaller), 
and (b) naturally capture diminishing returns. Formally, 
\begin{equation}
    \label{eq:freq}
    \freqLong \triangleq \log_{10}\Big(|\occs_\x| + 1\Big),
\end{equation}
where we add one to ensure that the logarithm is well-defined in the absence of occurrences.

Third, for $\spread(\cdot)$, to capture how regularly the occurrences are spread, we model the distribution of the \textit{gaps} $\gaps_x = (\gap_1, \dots, \gap_{|\uniqueOccs_\x| - 1})$ in a principled way via Shannon entropy~\cite{cover2012elements}. 
Entropy measures the amount of randomness in a distribution, and when the distribution has a fixed number of outcomes in its support, this essentially means it measures the distribution's uniformity.
The gaps $\gaps_\x$ between occurrences (\S~\ref{subsubsec:prelim-general}) normalized by the interval width $|\occInterval_\x|$ define a probability mass function with entropy 
\begin{equation}
    \label{eq:entropy}
    \entropy(\gaps_\x) \triangleq - \textstyle \sum_{\gap_i \in \gaps_x} \frac{\gap_i}{\occIntervalWidth}\log\frac{\gap_i}{\occIntervalWidth}.
\end{equation}

As is standard in information theory, we define $0 \log 0 = 0$ (e.g., multiple occurrences at the same time giving $\gap_i = 0$), and we use $\log$ base 2. In order to remove the dependency on the number of occurrences (since this is captured by $\freq$) and make this a measure of spread, we normalize by the maximum entropy $\log(|\gaps_\x|)$ \cite{cover2012elements}: 
\begin{equation}
    \label{eq:spread}
    \spreadLong \triangleq
    \begin{cases}
        \frac{\entropy(\gaps_\x)}{\log |\gaps_\x|} + 1, & |\gaps_\x| > 1 \\
        1, & |\gaps_\x| \in \{0, 1\}
    \end{cases},
\end{equation}
where we add one since entropy is 0 if there are 0 or 1 gaps.

Since the terms defined above have different units, we combine them into one function as follows:  
\begin{equation}
    \label{eq:persistence}
    \persistence\big(\pattern; \interval\big) \triangleq \widthLong^\wExp \freqLong^\fExp \spreadLong^\sExp,
\end{equation}
where the finite exponents $\wExp$, $\fExp$, $\sExp \in (0,\infty)$
can be used for assigning 
different weights to the components, depending on the goals of a particular application (e.g., $\sExp > 1$ can help events with low frequency but high regularity to be discovered). 

\subsection{Theoretical Analysis}
\label{subsec:measure-theory}
\noindent We now show that Eq.~\eqref{eq:persistence} is a principled persistence measure that satisfies all axioms and properties. 
We also discuss a recently-proposed persistence heuristic, which violates key axioms. 

\begin{lemma}
    \label{lemma:lower-bounds}
    If $x$ has occurred at least once, $|\occs_\x| > 0$, then our proposed measure is positive: $\persistenceLong > 0$, $\forall \ \wExp, \fExp, \sExp \in (0,\infty)$.
\end{lemma}

\begin{proof}
    $\widthLong > 0$ since $\intervalWidth + 1 \geq \occIntervalWidth + 1 > 0$. 
    Since 
    $|\occs_\x| > 0$, $\freqLong > 0$. 
    If $|\occs_\x| \in \{1, 2\}$, then $|\gaps_\x| \in \{0, 1\}$, so $\spreadLong = 1 > 0$ by definition. If $|\occs_\x| > 2$, $\log|\gaps_\x| > 0$ and the well-known property of entropy, $\entropy(\gaps_\x) \geq 0$ \cite{cover2012elements}, implies $\spreadLong > 0$. For $\wExp, \fExp, \sExp > 0$, each component is still greater than 0 and so is their product $\persistenceLong$.
\end{proof}

\begin{theorem}
    \label{thm:is-persistence-measure}
    $\persistenceLong$ as defined in Eq.~\eqref{eq:persistence} satisfies all the axioms, and thus is a persistence measure.
\end{theorem}

\begin{proof}
    We show that $\persistenceLong$ is a persistence measure by proving that it satisfies Axioms \axiomOne-\axiomFour. Since $\alpha, \beta, \gamma \in (0,\infty)$ do not affect the results, we omit them from the proofs for readability. 
    
     $\bullet$ Axiom \textbf{\axiomOne.} If $|\occs_\x| = 0$, then $\persistenceLong = 0 \cdot 0 \cdot 1 = 0$.
    If $|\occs_\x| > 0$, then $\persistenceLong > 0$, by Lemma \ref{lemma:lower-bounds}.
    Therefore, since $|\occs_\x| \geq 0$, if $\persistenceLong = 0$, then $|\occs_\x| = 0$.
    
     $\bullet$ Axiom \textbf{\axiomTwo.} Since $\occIntervalWidth \leq \intervalWidth$, $\widthLong \leq 1$. 
     Also $\entropy(\gaps_\x) \leq \log(|\gaps_\x|)$ \cite{cover2012elements}, so $\spreadLong \leq 2$. Thus, these bounded terms do not affect the limit, and $\textstyle \lim_{|\uniqueOccs_\x| \rightarrow \infty} \log_{10}(|\occs_\x| + 1) = \infty$. 
     Thus, unaffected by $\wExp$, $\fExp$, $\sExp \in (0,\infty)$, 
     $\textstyle \lim_{|\uniqueOccs_\x| \rightarrow \infty} \persistenceLong = \infty$.
    
     $\bullet$  Axiom \textbf{\axiomThree.} First, $\width(\pattern', \interval) = \frac{\lastOcc + c - \firstOcc - c + 1}{\intervalEnd - \intervalStart + 1} = \frac{\lastOcc - \firstOcc + 1}{\intervalEnd - \intervalStart + 1} = \widthLong$. 
    Also, $|\occs_{\x'}| = |\occs_\x|$, so $\freq(\pattern'; \interval) = \freqLong$. Each new gap $\gap_i' = t_{i+1} + c - t_i - c = t_{i+1} - t_i = \gap_i$ (\S~\ref{subsubsec:prelim-general}), so shifting the occurrences does not change their spread: $\spread(\pattern'; \interval) = \spreadLong$. 
    As a result, the persistence remains the same.

     $\bullet$  Axiom \textbf{\axiomFour.} By definition, $t_s \leq t_s' \leq t_f$ and $t_e \geq t_e' \geq t_l$, 
     so $\lastOcc - \firstOcc \leq \intervalEnd' - \intervalStart' \leq \intervalEnd - \intervalStart$. This implies that $\widthLong \leq \width(\pattern;\occIntervalsmall) \leq \width(\pattern;\occInterval_\x)$, with equality iff $\firstOcc = \intervalStart$ and $\lastOcc = \intervalEnd$. Since all occurrences fall in $\occInterval_\x$ by definition, the frequency and spread terms remain the same upon shrinking the measurement interval. 
     Therefore, $\persistenceLong \leq \persistence(\pattern;\occIntervalsmall) \leq \persistence(\pattern;\occInterval_\x)$, with equality iff $\firstOcc = \intervalStart$ and $\lastOcc = \intervalEnd$. \qedhere
     
\end{proof}

\begin{theorem}
    \label{theorem:prop1}
    Our persistence measure $\persistenceLong$, Eq.~\eqref{eq:persistence}, satisfies the three desired properties \propOne-\propThree.
\end{theorem}

\begin{proof}
    We prove each property separately:
    
    $\bullet$  Property \textbf{\propOne.} 
    First, since $|\occs_{{\uniformx_1}^n}| = |\occs_{{\uniformx_2}^n}| = n$, $\freq\big({\uniformx_1}^n;\interval\big) = \freq\big({\uniformx_2}^n;\interval\big)$. Also, for 
    $|\gaps_{{\uniformx_1}^n}| = |\gaps_{{\uniformx_2}^n}| = n - 1$ uniform gaps, $\entropy(\gaps_{{\uniformx_1}^n}) = \entropy(\gaps_{{\uniformx_1}^n}) = \log(n - 1)$. Thus, $\spread\big({\uniformx_1}^n;\interval\big) = \spread\big({\uniformx_2}^n;\interval\big)$. 
    Since $|\occInterval_{{\uniformx_1}^n}| < |\occInterval_{{\uniformx_2}^n}|$ and $\intervalWidth$ is fixed, $\width({\uniformx_1}^n;\interval) < \width\big({\uniformx_2}^n;\intervalWidth\big)$. Therefore, $\persistence\big({\uniformx_1}^n;\interval\big) < \persistence\big({\uniformx_2}^n;\interval\big)$.

    {
    $\bullet$ Property \textbf{\propTwo.} We prove this by case. For $n = 0$, $\persistence(\uniformx^0;\interval) = 0 < \persistence(\uniformx^{k};\interval)$ by Axiom~\axiomOne. 
    For $n \geq 1$,
    $\width(\uniformx^n;\interval) = \width(\uniformx^{n+k};\interval)$, since $\occInterval_{\uniformx^n} = \occInterval_{\uniformx^{n+k}}$.
    If $n = 1$, then $\spread(\uniformx^1;\interval) = 1 \leq \spread(\uniformx^{1+k};\interval) \leq 2$ and $\freq\big(\uniformx^1;\interval\big) = \log_{10}(2) < \log_{10}(2+k) = \freq\big(\uniformx^{1+k};\interval\big)$.
    For $n > 1$, because of uniform gaps, $\entropy(\gaps_{\uniformx^n}) = \log(|\gaps_{\uniformx^n}|)$ and $\entropy(\gaps_{\uniformx^{n+k}}) = \log(|\gaps_{\uniformx^{n+k}}|)$ \cite{cover2012elements}.  
    Thus, $\spread(\uniformx^{n};\interval) = \spread\big(\uniformx^{n+k};\interval\big) = 2$. 
    Lastly, $\freq\big(\uniformx^n;\interval\big) < \freq(\uniformx^{n+k};\interval)$. 
    By combining the three terms, we prove the claim. }
    
     $\bullet$  Property \textbf{\propThree.} Since $|\occs_\x| = |\occs_{\uniformx^n}|$, $\freqLong = \freq\big(\uniformx^n;\interval\big)$. 
     Also, $\occInterval_\x = \occInterval_{\uniformx^n}$, so $\widthLong = \width\big(\uniformx^n;\interval\big)$. 
     Moreover, $\entropy(\gaps_\x) \leq \log(|\gaps_\x|)$, with equality iff the $\gap_i \in \gaps_\x$'s are uniform. 
     Thus, $\spreadLong \leq 2 = \spread\big(\uniformx^n;\interval\big)$, with equality iff the gaps $\gap_i \in \gaps_\x$'s are uniform, in which case $\gaps_\x = \gaps_{\uniformx_n}$. \qedhere

\end{proof}

\noindent \textbf{Persistence Heuristic in Data Streams~\cite{dai2016finding}.} 
A simple persistence heuristic was proposed in \cite{dai2016finding}: an item in a data stream is considered persistent if it occurs at least once in a large number of \textit{predefined}, equally-sized observation periods (intervals).
We show that this heuristic violates axioms \axiomTwo-\axiomFour\ via counter-examples. As the number of unique occurrences tends to infinity, the heuristic tends to the number of intervals,
not infinity (\axiomTwo\ violation). If an item has two occurrences in the same interval, they can be shifted such that one falls in a new period, in which case the heuristic 
grows (\axiomThree\ violation). 
Finally, for an item with two occurrences in different intervals, after 
shrinking $\interval$ towards $\occInterval_\x$ and re-dividing into the predefined number of intervals, they could still fall in different intervals. 
Thus persistence would not increase (\axiomFour\ violation).

\section{Proposed Method: \method}
\label{sec:method}

With our proposed persistence measure in Eq.~\eqref{eq:persistence}, and the definitions in \S~\ref{subsec:prelim}, we can now define the offline and streaming (online) versions of our problem precisely. 

\begin{problem}[Offline/Streaming Persistent Activity Mining]\label{problem:prob-formal}
Given an edge stream, $\stream$, a maximum duration $\maxDuration$, and a maximum snippet size $\maxSize$, 
\begin{itemize}
    \item {\emph{[}Offline\emph{]}} \textbf{output} the persistence $\persistence\big(\pattern;[\streamStart,\streamEnd]\big)$
    \item {\emph{[}Streaming\emph{]}} \textbf{maintain} the persistence $\persistence(\pattern;[\streamStart,\timestamp])$
\end{itemize}
 of every \emph{(}$\maxDuration, \maxSize$\emph{)}-activity snippet $\pattern$ observed in the whole stream. 
\end{problem}

We next introduce our offline and streaming algorithms, \methodOffline and \methodOnline, to solve this problem. 
We first discuss extracting all $(\maxDuration,\maxSize)$-snippets from a stream $\stream$, followed by the details of each variant.
We give the main steps for \method, which takes as a parameter which variant to use, in Algorithm~\ref{alg:method}, and detailed pseudocode for reproducibility in \S~\ref{subsec:detailed-pseudocode}.

\vspace{0.1cm}
\label{subsubsec:pattern-extraction}
\noindent \textbf{Activity Snippet Extraction (line 2).} \label{subsubsec:alg-extract-patterns} Each time a new edge-update $\newUpdate$ arrives (at time $\timestamp$), the procedure \textsc{ExtractNewSnippets} is called. The procedure maintains the window $\window$, which consists of all updates $\update$ 
that occurred within the last $\maxDuration$ time units, and removes all others. 
It then adds the new update $\newUpdate$ (which is 0 time-units in the past). In addition to maintaining the window, the procedure extracts all valid snippets from the window. A valid snippet must be connected, have duration $\duration \leq \maxDuration$, and size $\patternSize \leq \maxSize$.  Since all stale updates have been removed from $\window$, $\maxDuration$ is already enforced. Any new snippet instance must contain $\newUpdate$.  
The singleton snippet containing just $\newUpdate$ is created, and then snippets of size $\patternSize = 2, \dots, \maxSize$ are constructed smallest to largest, such that the nodes in each snippet are connected.

\vspace{0.1cm}
\noindent \textbf{Offline Algorithm (lines 4 and 8).}
\label{subsubsec:offline}
For the offline version of Problem~\ref{problem:prob-formal}, \methodOffline maintains the set of occurrences of each activity snippet extracted from the stream. Then, when the end of the stream is reached, it computes and outputs the persistence of each snippet $\pattern$ in $[\streamStart, \streamEnd]$ with Eq.~\eqref{eq:persistence} as $\persistence(\pattern;[\streamStart, \streamEnd])$.

\begin{algorithm}[t]
	\caption{\method($\stream$, $\maxDuration$, $\maxSize$, $\view$, $\wExp$, $\fExp$, $\sExp$, \textsc{variant})}
	\label{alg:method}
	\begin{algorithmic}[1]
	    \small
	    \Statex \textbf{Input}: Stream $\stream$, max snippet duration $\maxDuration$ and size $\maxSize$, view $\view$, persistence exponents $\wExp$, $\fExp$, $\sExp$, the variant (\methodOffline/\methodOnline).
	    \While{$\newUpdate \in \stream$} \Comment{\textcolor{gray}{While there is a new update in the stream}}
	       \For{$\pattern \in \Call{ExtractNewSnippets}{\window, \newUpdate, \timestamp, \view}$}
	            \If{\textsc{variant} is \methodOffline}
	               \State{Add the occurrence $\timestamp$ of $\pattern$ to $\occs_\pattern$.}
	            \Else
	               \State Update $\persistence(\pattern;[\streamStart;t])$ incrementally.
	            \EndIf
	        \EndFor
	    \EndWhile
	    \If{\textsc{variant} is \methodOffline}
	        \State Compute $\persistence(\pattern;[\streamStart, \streamEnd])$ for each $\pattern$ observed.
	    \EndIf
	\end{algorithmic}
	\vspace{-0.1cm}
\end{algorithm}

\vspace{0.1cm}
\noindent \textbf{Streaming Algorithm (line 6).}  
\label{subsubsec:online}
For the online version of Problem~\ref{problem:prob-formal}, there are two cases to handle: (1) update the persistence of snippet $\pattern$ when it occurs at time $\timestamp$, and (2) return the correct persistence of a snippet $\pattern$ if it is queried at any other time $\timestamp$. We maintain in memory a constant amount of information on each snippet: the total number of its occurrences, $|\occs_\x|$, the number of gaps between \vspace{-0.1cm}occurrences, $|\gaps_\x|$, the time of its first and last occurrences, $\firstOcc$, $\lastOcc$, and its persistence when it last occurred, $\persistence(\pattern;[\streamStart,\lastOcc])$. 

$\bullet$ \textbf{Case 1. Updating Persistence Upon Occurrence.} 
Since there is a new occurrence at the current time $\timestamp$ and we maintain the first occurrence of each snippet, we can compute the width function $\widthLong^\wExp = \left(\frac{|[\firstOcc, \timestamp]_\pattern| + 1}{|[\streamStart,\timestamp]| + 1}\right)^\wExp$ from Eq.~\eqref{eq:width}. Since we maintain the number of occurrences, $|\occs_\pattern|$, we can obtain the frequency  $\freq\big(\pattern;[\streamStart;\timestamp]\big)^\fExp = \log_{10}\big(|\occs_\pattern| + 1\big)^\fExp$ from Eq.~\eqref{eq:freq}. By maintaining the number of gaps between occurrences, we know whether $|\gaps_\pattern| \in \{0, 1\}$ (in Eq.~\eqref{eq:spread}) and can compute $\log|\gaps_\pattern|$. In order to compute $\entropy(\gaps_\pattern)$---Eq.~\eqref{eq:entropy}---we show how the entropy 
of the distribution induced by a snippet $\pattern$'s gap widths
can be computed incrementally, as new occurrences create new gaps. 
Specifically, the entropy 
when a new gap is formed can be computed from (1)~the previous entropy,
(2)~the previous normalizing constant $Z$,
and (3)~the new gap $\gap_{n+1}$.
This is stated in Thm.~\ref{thm:entropy-incremental} and proved via the following two Lemmas. 

Let $p_i \triangleq \frac{\gap_i}{Z}$, where $Z = \sum_{i=1}^n \gap_i$, be a probability mass function on a set of $n$ gaps, $\gaps_\x$, induced by normalizing each gap by $Z$.
\begin{lemma}
    \label{lemma:entropy-n}
    The entropy $\entropy(p) = \log Z - \frac{1}{Z}\sum_{i=1}^n \gap_i\log \gap_i.$
\end{lemma}

\begin{lemma}
    \label{lemma:entropy-n-1}
    For the new set of gaps ${\gaps'}_\x = \gaps_\x \cup \{\gap_{n+1}\}$, and normalizing constant $Z' = Z + \gap_{n+1}$, the corresponding pmf $p'$ has entropy $\entropy(p') = \frac{Z}{Z'}\log Z' -\frac{\gap_{n+1}}{Z'}\log\frac{\gap_{n+1}}{Z'} - \frac{1}{Z'}\sum_{i=1}^n \gap_i\log \gap_i$.
\end{lemma}

\begin{proof}
    Both lemmas can be derived by algebraically expanding the corresponding entropy definition and using $Z = \sum_{i=1}^n \gap_i$. 
\end{proof}

\begin{theorem}
    \label{thm:entropy-incremental}
    $\entropy(p') = \entropy(p) + \frac{Z}{Z'}\log Z' -\log Z -\frac{\gap_{n+1}}{Z'}\log\frac{\gap_{n+1}}{Z'} + \left(\frac{1}{Z}-\frac{1}{Z'}\right)(\log Z - \entropy(p))Z$, where $\entropy(p)$ is the entropy of $p$, $Z$ is the normalizing constant for $p$, and $\gap_{n+1}$ is the new gap.
\end{theorem}

\begin{proof}
    Let $\Delta = \entropy(p') - \entropy(p)$ be the change in entropy and  $X = \sum_{i=1}^n \gap_i\log \gap_i$.
    By Lemmas~\ref{lemma:entropy-n}-\ref{lemma:entropy-n-1}, we obtain $\Delta = \frac{Z}{Z'}\log Z' -\log Z - \frac{\gap_{n+1}}{Z'}\log\frac{\gap_{n+1}}{Z'} + \left(\frac{1}{Z}-\frac{1}{Z'}\right)X$. 
    By Lemma~\ref{lemma:entropy-n}, 
    $X = (\log Z - \entropy(p))Z$. 
    Plugging, this into $\entropy(p') = \entropy(p) + \Delta$ completes the proof.
\end{proof}

The new gap $\gap_{n+1}$ is the time from the previous occurrence of $\pattern$ (which we know since we maintain $\lastOcc$) to the current time $\timestamp$. The normalizing constant $Z$ is the time from the first occurrence ($\firstOcc$) to the previous occurrence ($\lastOcc$), both of which we maintain.

$\bullet$ \textbf{Case 2. Querying Persistence 
Without a New Occurrence.}
Without a new occurrence, the frequency
$\freq(\pattern;[\streamStart,\timestamp])^\fExp$ and 
spread 
$\spread(\pattern;[\streamStart,$$\timestamp])^\sExp$ do not change. From Eq.~\eqref{eq:width}, we compute the width function,   $\width(\pattern;[\streamStart,\timestamp])^\wExp$, by replacing $|[\streamStart,\lastOcc]|$ with $|[\streamStart,\timestamp]|$ in the denominator.

\section{Evaluation}
\label{sec:eval}
We investigate the following research questions across multiple real networks (\S~\ref{subsec:data}):

\begin{itemize}
    \item \rOne{} What does the relationship between frequency and persistence reveal about activity and networks? 
    \item \rTwo{} Can \methodOnline{} find anomalies in real-time?
    \item \rThree{} Can \methodOnline{} process updates to a network at least as quickly as they arrive, and how does \methodOffline{} scale with the number of edge updates in the stream?
\end{itemize}

In \S~\ref{subsec:p_vs_f} we investigate a core value of measuring persistence, which is the relationship between frequency and persistence. 
Specifically, we find that activity snippets with high persistence, but low frequency
tend to correspond to \emph{subtle} yet regular activity, which would be missed by measuring frequency alone. On the other hand, snippets with high frequency, but low persistence 
tend to be \emph{bursty}. 
In \S~\ref{subsec:realtime}, we show how we can make analogous insights in real-time with \methodOnline{}, and \textbf{accurately} find both subtle and bursty anomalies right when they occur. Finally, in \S~\ref{subsec:scale}, we evaluate the maximum duration $\maxDuration$ and maximum size $\maxSize$ parameters, and demonstrate that \methodOnline{} processes edges in each stream 10K to 360K times faster than the rate of that stream, and \methodOffline{} scales linearly with the number of edge-updates processed.

\subsection{Data}
\label{subsec:data}

\begin{table}[b!]
\vspace{-0.2cm}
\small 
    \centering
    \caption{Description of edge streams: number of edge-updates, number of nodes, whether it has edge deletions, number of unique edges, number of edge types, and average rate of the stream in updates/sec.}
    \label{table:stats}
    \vspace{-0.35cm}
    \resizebox{\columnwidth}{!}{
        \begin{tabular}{l r r r r r r r}
        \toprule
          & $\streamLen$ & $\numNodes$ & Del & $\numEdges$ & $\edges$ types & rate ($\update$/sec)  \\
        \midrule
        \euemail & 332,334 & 986 & N & 24,929 & 1 & 0.04\\
        
        \columbus & 534,998 & 74 & Y & 2,951 & 1 & 0.02\\

        \reddit & 858,488 & 67,180 & N & 339,643 & 1 & 0.02\\
        
        Darpa IP & 4,554,344 & 25,525 & N & 68,910 & 1 & 78.5\\

        Boston Bike & 17,421,182 & 476 & Y & 81,508 & 1 & 0.05\\

        Chicago Bike & 33,331,104 & 712 & N & 171,651 & 1 & 0.25 \\
        
        Stackoverflow & 63,497,050 & 2,601,977 & N & 36,233,450 & 3 & 0.27 \\

        NYC Taxi & 3,077,990,404 & 265 & Y & 60,750 & 1 & 9.29\\
        \bottomrule
        \end{tabular}
        }
\end{table}

We utilize a diverse set of evolving networks (Tab.~\ref{table:stats}), including communication, transportation, computer, and social networks. 
\begin{itemize}[leftmargin=*]
    \item The communication network \euemail has timestamped edges denoting emails sent within a European research institute \cite{snapnets}. 
    \item \columbus, \chicago, and \boston are networks encoding bike trips made in the bike-share systems in those three cities \cite{bikedatasets}. Each node is a bike station. When a bike leaves one station for another, an edge is inserted into the network; when the bike arrives at its destination, it is deleted. Thus, the networks capture \emph{en route} bike-trips.
    \item Similarly, \NYC \cite{taxidataset} captures \emph{en route} taxi trips in New York City, but in this case, nodes are city zones rather than stations.
    \item Edges in the social network \reddit correspond to timestamped references between subreddits (topical discussion boards) \cite{snapnets}. 
    \item \darpa \cite{lippmann1999results} is an IP-IP network, where both normal network traffic and malicious attack traffic is present. Edges denote interactions between computers in the network.
    \item \stackoverflow \cite{snapnets} consists of interactions among users on the website Stackoverflow. The interactions are between users, and can be answers to questions, comments on questions, or comments on answers.
\end{itemize}

\subsection{\rOne{}: Persistent vs. Frequent}
\label{subsec:p_vs_f}

{The relationship between frequency and persistence allows for a more complete view of activity in networks.}
We show this via a Persistence vs. Frequency, or \textbf{\pf}, plot of activity snippets, with snippet frequency, $|\occs_\pattern|$, in log-scale on the \yaxis, and persistence $\persistence(\pattern;[\streamStart,\streamEnd])$ on the \xaxis. 
Activity snippets with unusually high persistence relative to their frequency fall towards the lower-right of the plot. Those with unusually high frequency relative to their persistence fall towards the upper-left. We call the former type of activity snippet \textbf{subtly persistent}, and the latter \textbf{bursty} (whether a single burst or periodic bursts). We analyze persistent, subtly persistent, and bursty activity snippets in several 
networks via \pf plots, and show what they correspond to in each context.
{To see what the baseline of frequency would give, the variation of points along the \yaxis only can be studied.}

\subsubsection{Transportation Networks}
\label{subsubsec:transportation-networks}

\textbf{Setup.} We generate \pf plots for bike trips in Boston, MA, Columbus, OH, and Taxi trips in New York City (NYC) for the two weeks surrounding Hurricane Sandy (10/22/12-11/05/12). We use view $\view = \idView$. Since once a trip is complete the rider leaves the bike/taxi, sequences of edges are not linked. Thus, we set $\maxSize = 1$ (and $\maxDuration=1$), so that each activity snippet corresponds to a single bike/taxi trip. Next, in social networks, we use $\maxSize > 1$. 
We set $\wExp, \fExp$ and $\sExp$ for visual clarity, and discuss how in the supplement on reproducibility (\S~\ref{subsec:setting-exponents}).

\vspace{0.1cm}
\noindent  \textbf{Results.} Results for \boston and \columbus are in Fig.~\ref{fig:bikes}. The most persistent bike trips (black) are both frequent and persistent. A representative timeline from each, where each tick is an occurrence of the snippet, is shown in the lower-right of each plot. In Boston, the most persistent trip is from Massachusetts Ave. in front of MIT, up the street to outside 
the Central subway station. This is a reasonable route for commuters to bike from MIT to the subway. In both networks, bursts (orange) reveal new bike stations opening. 
The new station in Boston opened right in front of MIT and became immediately popular, which led to a large burst of activity.
With the station now established, its persistence should grow over time. 
For comparison, we show a representative point from the middle of the Columbus \pf plot, and a trip in Boston that is more persistent than expected given its frequency (blue).

\begin{figure}[t!]
    \centering
    \includegraphics[width=\linewidth]{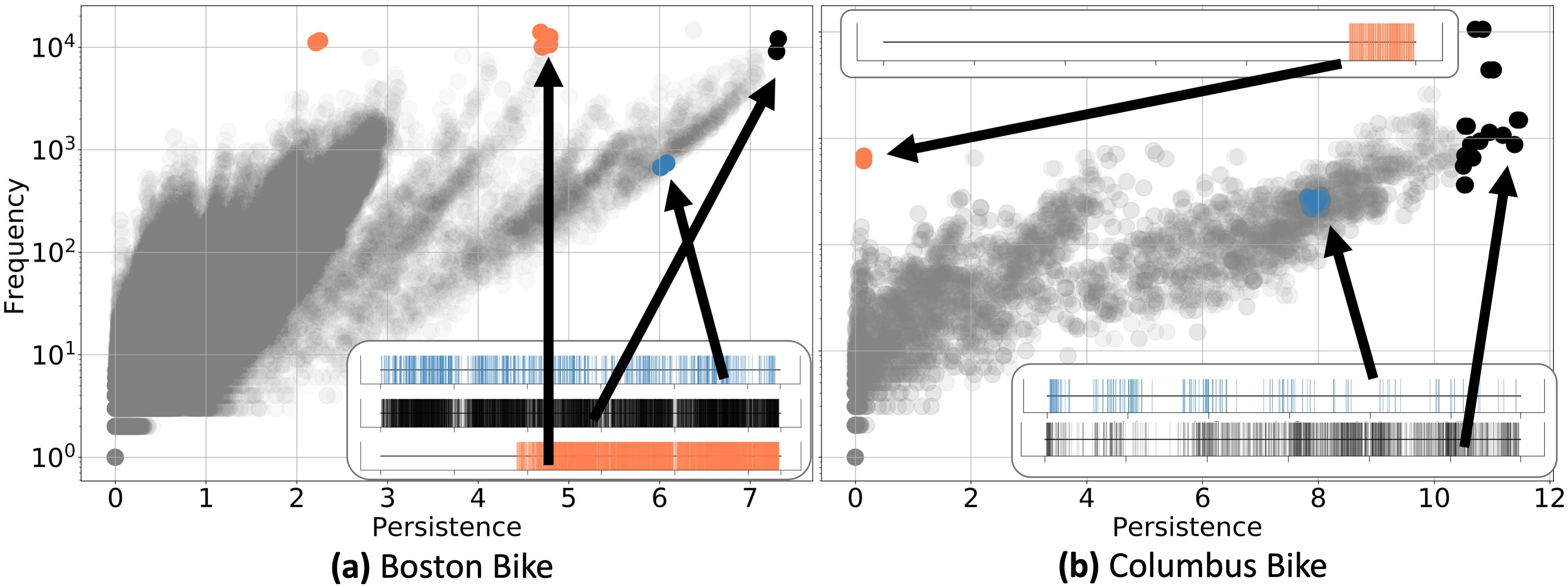}
    \vspace{-0.45cm}
    \caption{Boston (a) and Columbus (b) bike networks. Representative timelines from various parts of the \pf plots demonstrate how persistent, bursty, and subtle activity can be identified.}
    \label{fig:bikes}
    \vspace{-0.2cm}
\end{figure}

\begin{figure}[t!]
    \centering
    \includegraphics[width=\linewidth]{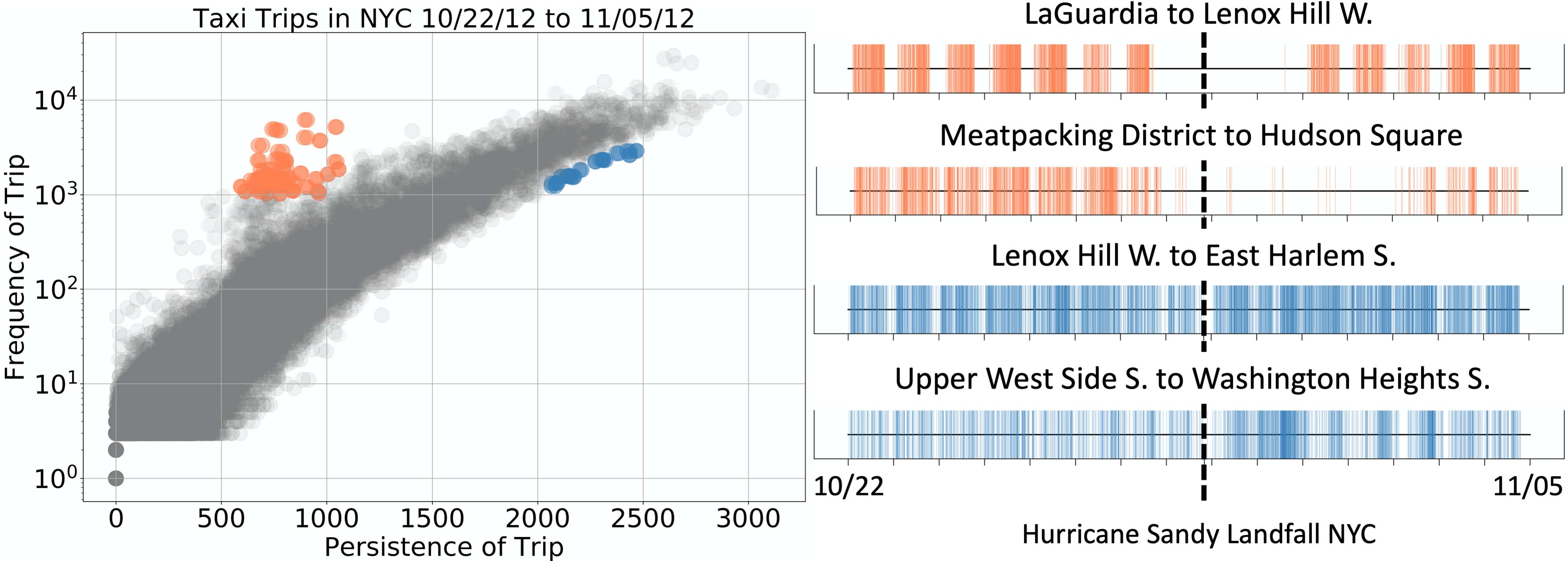}
    \vspace{-0.5cm}
    \caption{Taxi trips in NYC with high frequency but low persistence reveal neighborhoods that were brought to a standstill by Hurricane Sandy, while more persistent trips reveal those that were not.}
    \label{fig:nyc-sandy-f-vs-p}
    \vspace{-0.4cm}
\end{figure}

In \NYC, the bursty and subtly persistent snippets (Taxi trips) reveal which neighborhoods in NYC were most affected by Hurricane Sandy (Fig.~\ref{fig:nyc-sandy-f-vs-p}). The hurricane made landfall in NYC around 8pm EST on 10/29 (shown with a dashed line). The bursty anomalies (orange) all correspond to neighborhoods that were brought to a stand-still during the storm. For example, the first trip shown is from LaGuardia airport to the Lenox Hill neighborhood. Flights were canceled \emph{en masse} at LaGuardia prior to the hurricane, interrupting trips even 
before the storm's landfall. The second timeline shows a trip from the Meatpacking District to Hudson Square. Both neighborhoods are on the coast of Manhattan Island, and experienced severe flooding. The decrease in trips prior to the storm is likely due to businesses closing in preparation. Many of the subtly persistent trips (blue) correspond to neighborhoods that were resilient despite the storm. The first example is Lenox Hill and East Harlem. While these neighborhoods border the East River, they extend inland, allowing plenty of roads for taxis to continue service on. Lenox Hill hospital was also the subject of a study of the clinical response to Hurricane Sandy \cite{walczyszyn2019battling}. Possibly, the continuity in taxi trips was due to visits to the hospital. Similarly, the Upper West Side extends into inner-Manhattan, and Washington Heights is a significant distance north of where the hurricane made landfall.

\textbf{Takeaways on Transportation.} We find that subtly persistent activity often reveals commutes, because of their high regularity. This has applications in city planning, since commute routes can be good options for introducing ride-share programs. We also find that bursty activity often reveals major changes in the real world, such as severe weather disturbances or new routes becoming possible.

\subsubsection{Social Networks}
\label{subsubsec:social-networks}

\textbf{Setup.} We use the social network \stackoverflow, where users interact by answering technical questions, commenting on questions, and commenting on answers. We first analyze a 3-month interval. For this analysis, we use the view $\view = \idView$, set $\maxSize = 3$, and $\maxDuration = 1$hr. Secondly, we analyze the entire dataset using $\view = \orderView$, with $\maxSize = 3$ and $\maxDuration = 15$min. In the latter scenario, the snippets we analyze are a superset of the 3-node, 3-edge temporal motifs studied in \cite{paranjape2017motifs}. We plot, in orange, snippets involving users commenting on questions, in black, those involving users commenting on answers, and in blue, answering questions.

\vspace{0.1cm}
\noindent  \textbf{Results.} Results are shown in Fig.~\ref{fig:stackoverflow}. The timelines in (a) are for activity snippets among \emph{specific} users. The bursty anomaly (orange) reveals users \texttt{u72603} and \texttt{u82199} (anonymized ids) commenting on \texttt{u82199}'s answer 36 times over the course of 1hr in the 3-month interval, then never interacting again. The timeline shown is zoomed-in on for clarity. The most persistent snippets (one shown in black) are users commenting on their own answers, suggesting that it is unusual for the same two distinct users to interact persistently over time. On the other hand, the subtly persistent anomaly (blue) reveals \texttt{\small u1950} regularly answering \texttt{\small u55747}'s questions. We show the persistent, bursty, and subtly anomalous snippets in (i)-(iii). We give the \pf plot used to identify these snippets in \S~\ref{subsec:pf-supplement}.

The right plot uses the view $\view = \orderView$, where any users can form an occurrence of the activity snippet. Remarkably, we find that activity involving the three interaction types (commenting on questions, commenting on answers, and answering questions) fall in distinct places on the plot. Activity involving comments on questions tends to fall in the bursty region of the plot, while answers are the most subtly persistent. 
To better understand the phenomenon, we zoom in on the region shown in the upper-left.
We observe that most activity snippets in this box target a single user (i.e., all discussion is directed towards that user), and involve at least two distinct users. We give examples in (iv)-(v) of Fig.~\ref{fig:stackoverflow}. In (iv), a user comments on another's post, the second user responds, and the original user comments again. In (v), the second user responds twice to the original comment. These snippets capture natural, technical discussions among users. Remarkably, these discussions center around questions (orange) and answers (black) at roughly equal frequencies, but they occur with higher persistence for answers. We conjecture that this is because, when a question is asked, it triggers a flurry of comments on it; but once the question is answered, these tend to die out. On the other hand, comments on \emph{answers} persist, because new users may have the same question, and have follow-up questions, even months after the question has been answered. {Since both have indistinguishable frequencies ($10^5$-$10^7$), this subtle difference in behavior is not revealed by frequency alone.}

\textbf{Takeaways on Social Networks.} We find that subtly persistence activity can reveal users who do not interact often, but do \emph{regularly}---i.e., \emph{similar} users who are missed when only looking at how many times they interact. We also find that some types of posts promote continual activity, while others trigger bursts.

\begin{figure}[t!]
    \centering
    \includegraphics[width=\linewidth]{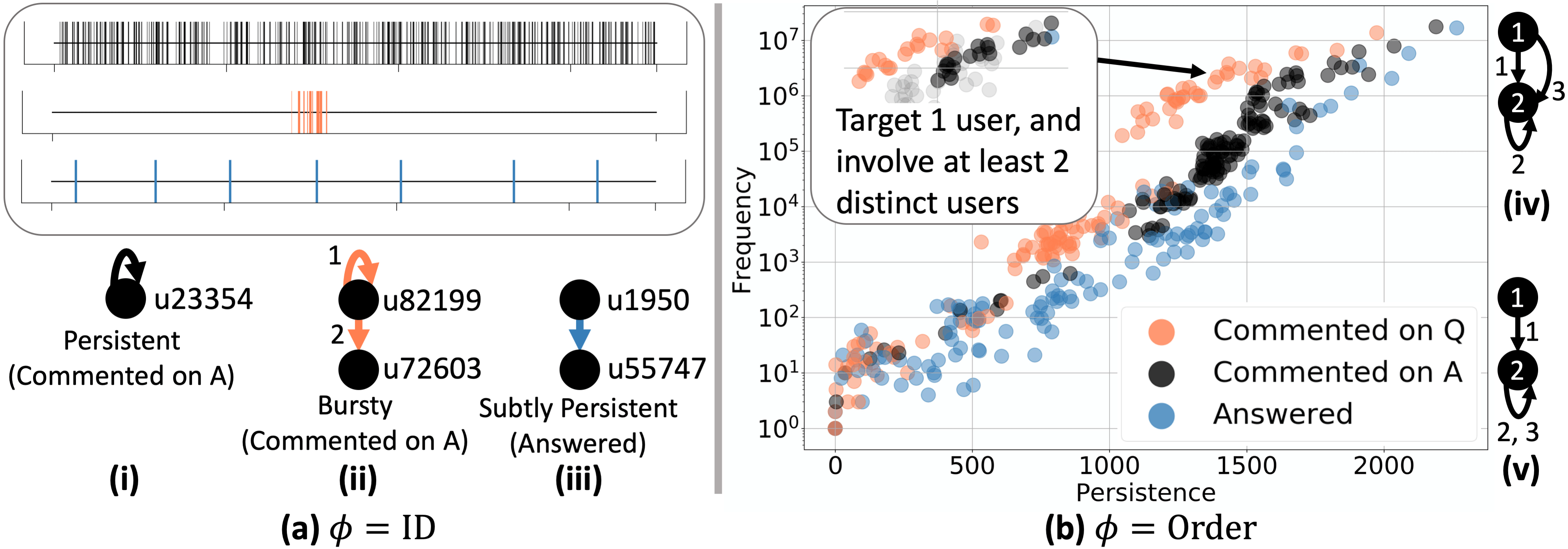}
    \vspace{-0.5cm}
    \caption{Stackoverflow Analysis. (a): $\view = \idView$. Persistent snippets usually correspond to users commenting on their own answers. The bursty snippet reveals 36 back-and-forths between two users in 1hr. The subtly persistent snippet reveals user \texttt{\small u1950} regularly answering user \texttt{\small u55747}'s questions. (b): $\view = \orderView$. Discussions targeting one user's content and involving at least two users, occur with similar frequency for comments on Qs and As, but comments on Qs are burstier---a distinction not captured by frequency alone.}
    \label{fig:stackoverflow}
    \vspace{-0.35cm}
\end{figure}

\subsection{\rTwo{}: Anomalies in Real Time}
\label{subsec:realtime}
We next demonstrate that bursty and subtly persistent snippets can be identified in real time, right when they occur. 

\setlength{\tabcolsep}{2pt}

\subsubsection{Generating Anomaly Scores}
\label{subsubsec:gen-anomaly}

{To identify anomalies automatically in real-time (without visually inspecting \pf plots as in \S~\ref{subsec:p_vs_f}), we use the following process.
When an activity snippet $\pattern$ appears in the stream at time $\timestamp$, we generate a 2D point $<$frequency, persistence$>$, or $[\freq(\pattern;[\streamStart,\timestamp])$, $\persistence(\pattern;[\streamStart,\timestamp])]$, corresponding to the dimensions of a \pf plot. Then, any streaming anomaly detection method can be applied. We use the Random Cut Forest (RCF) method \cite{guha2016robust}, which 
gives a real-valued anomaly score for each 
point in a stream. 
For implementation details, see \S~\ref{subsec:rcf-supplement}. }

\subsubsection{Real-time Qualitative Anomaly Detection (\pf)}
\label{subsubsec:realtime-PvF}

\textbf{Setup.} We analyze two networks: \reddit and \NYC trips from the first two months of 2019. For each, we choose ground-truth snippets by choosing a subtly persistent anomaly, a bursty anomaly, and a snippet corresponding to a point from the middle of the \pf plot. We use \methodOnline with $\maxSize = 1$, and set $(\wExp, \fExp, \sExp)$ values, which we give in \S~\ref{subsec:param-ref}. At each time $\timestamp$, we compute the anomaly scores of each ground-truth snippet (\S~\ref{subsubsec:gen-anomaly}), and
compare the scores to the median and standard deviation of the anomaly scores of all points seen to that point. We label an activity snippet at time $\timestamp$ as a level 1, 2, or 3 anomaly if it is 1 to 2, 2 to 3, or 3+ standard deviations above the median score respectively.

\vspace{0.1cm}
\noindent \textbf{Results.} The results in Fig.~\ref{fig:real-time-anom} for \reddit (left) and \NYC (right) show level 1, 2, and 3 anomalies colored green, orange, and maroon. 

The subtly persistent snippet in \reddit corresponds to the subreddit \cig referencing \poker. This is the 3rd most persistent snippet overall, but only the 252nd most frequent. References between these subreddits occur remarkably regularly. Upon investigation, we found that \poker was formed by a popular user in \cig. We conjecture that the snippet corresponds to this user, or their followers, promoting the content of the other subreddit. 
The bursty anomaly corresponds to \odds referencing \hockey, presumably picking winners for each night's hockey game. The bursts align with hockey seasons. \methodOnline consistently identifies this activity snippet as anomalous. Furthermore, the anomaly score decreases as expected over time, since as the bursts return yearly during hockey season, the snippet becomes more persistent. 
The third snippet is neither bursty nor regular, and it is correctly not flagged as an anomaly.

In NYC, 
the subtly persistent anomaly reveals a taxi trip from Kew Gardens Hills in Queens, to Manhattan, near the United Nations building. The taxi trip is repeated every day shortly after midnight, and is almost never taken at any other time. The nature of the trip is unknown, but surprisingly regular. The score, as expected, grows over time, as the continued regularity increases anomalousness. 
The bursty anomaly captures taxi trips departing from and arriving at the zone containing the NYC Taxi \& Limousine Commission. The Taxi Commission inspects taxis and is open 5 days a week, which suggests that the bursts correspond to test-drives of taxis for inspection during business hours. \methodOnline consistently identifies these bursts. 
Again, the third snippet is neither bursty nor regular, and is not often flagged as an anomaly.

\begin{figure}[t!]
    \centering
    \includegraphics[width=\linewidth]{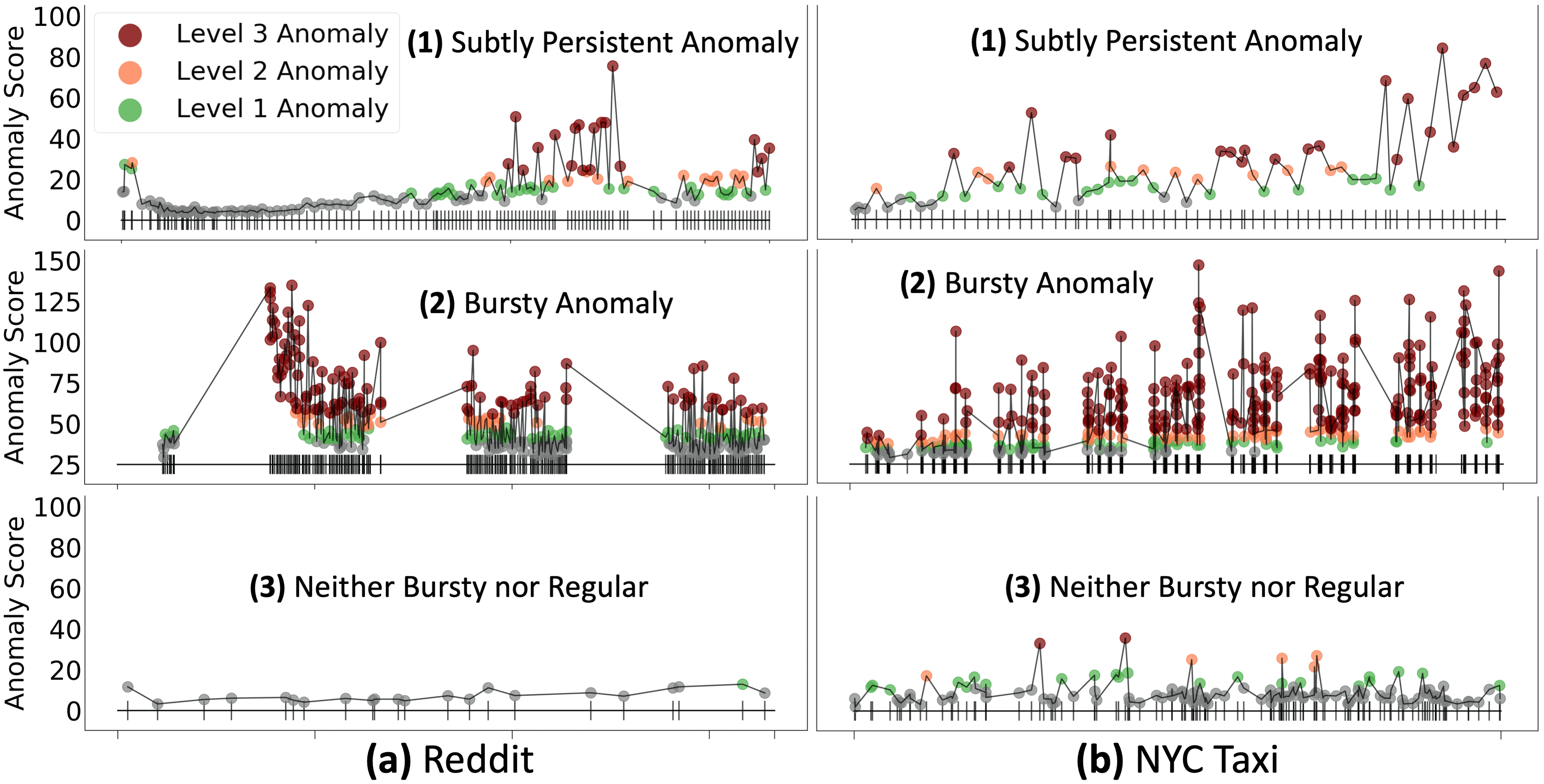}
    \vspace{-0.7cm}
    \caption{With \method{}, we are able to identify anomalies in \emph{real-time}. Not only can we find bursty anomalies, but also subtly \emph{persistent} anomalies: those that occur regularly and continually, but with frequency too low to be discovered by aggregate count alone. Anomaly levels capture how anomalous an occurrence is \S~\ref{subsubsec:realtime-PvF}.}
    \label{fig:real-time-anom}
    \vspace{-0.2cm}
\end{figure}

\setlength{\tabcolsep}{2pt}
\begin{table}[t]
\vspace{-0.1cm}
\centering
\scriptsize 
	\caption{{Results for identifying subtly persistent and bursty anomalies. Statistically significant results are marked with an ``$\ast$''. \methodOnline outperforms all baselines at identifying subtly persistent anomalies, which is not a well-studied problem. \methodOnline also performs competitively with baselines on bursty anomaly detection, leading to the best overall performance (avg AUC).} 
    }
	\vspace{-0.35cm}
	\label{table:quant-anomaly-detection}
	\resizebox{\columnwidth}{!}{
	\begin{tabular}{@{}llccccccccc}
		\toprule
		\emph{} & \emph{Metric} & \freqBaseline & \sedanspot \cite{eswaran2018sedanspot} & \midas \cite{bhatia2020midas} & DS \cite{dai2016finding} & \methodOnline \\
		\toprule
		\multirow{4}{*}{{\rotatebox{90}{\bf Subtle}}}
		& AUC & $0.8325 \scriptsize{\pm 0.02}$ & $0.4519 \scriptsize{\pm 0.01}$ & $0.4520 \scriptsize{\pm 0.02}$ & $0.7435 \scriptsize{\pm 0.03}$ & $\mathbf{0.9309}\scriptsize{\pm 0.00}\ast$ \\
		& F1@100 & $0.0505 \scriptsize{\pm 0.01}$ & $0.0001 \scriptsize{\pm 0.00}$ & $0.0000 \scriptsize{\pm 0.00}$ & $0.0076 \scriptsize{\pm 0.00}$ & $\mathbf{0.0508}\scriptsize{\pm 0.01}$\\
		& F1@1K & $0.1812 \scriptsize{\pm 0.00}$ & $0.0035 \scriptsize{\pm 0.00}$ & $0.0003 \scriptsize{\pm 0.00}$ & $0.0378 \scriptsize{\pm 0.01}$ & $\mathbf{0.2580}\scriptsize{\pm 0.03}\ast$\\
		& F1@2K & $0.1572 \scriptsize{\pm 0.01}$ & $0.0098 \scriptsize{\pm 0.00}$ & $0.0002 \scriptsize{\pm 0.00}$ & $0.0561 \scriptsize{\pm 0.01}$ & $\mathbf{0.3292}\scriptsize{\pm 0.03}\ast$\\
        \midrule
        \multirow{4}{*}{{\rotatebox{90}{\bf Bursty}}}
        & AUC & $0.8450 \scriptsize{\pm 0.00}$ & $0.6390\scriptsize{\pm 0.00}$ & $\mathbf{0.9434}\scriptsize{\pm 0.00}\ast$ & $0.8632\scriptsize{\pm 0.00}$ & $0.8359\scriptsize{\pm 0.01}$\\
		& F1@500K & $\mathbf{0.3089}\scriptsize{\pm 0.00}\ast$ & $0.2745\scriptsize{\pm 0.00}$ & $0.3019\scriptsize{\pm 0.00}$ & $0.3063\scriptsize{\pm 0.00}$ & $0.2978\scriptsize{\pm 0.00}$ \\
		& F1@1M & $\mathbf{0.5351}\scriptsize{\pm 0.00}\ast$ & $0.4527\scriptsize{\pm 0.00}$ & $0.5274\scriptsize{\pm 0.00}$ & $0.5295\scriptsize{\pm 0.00}$ & $0.5169\scriptsize{\pm 0.00}$\\
		& F1@2M & $0.7184\scriptsize{\pm 0.00}$ & $0.6309\scriptsize{\pm 0.00}$ & $\mathbf{0.8378}\scriptsize{\pm 0.00}\ast$ & $0.8066\scriptsize{\pm 0.00}$ & $0.7770\scriptsize{\pm 0.01}$ \\
		\midrule
        & \textbf{Avg AUC} & 0.8388 (2) & 0.5455 (5) & 0.6977 (4) & 0.8034 (3) & \textbf{0.8834} (1) \\
  	    \bottomrule
	\end{tabular}
	}
\vspace{-0.3cm}
\end{table}

\subsubsection{Real-time Quantitative Anomaly Detection}
\label{subsubsec:quant-anom}

\textbf{Setup.} We quantitatively analyze \methodOnline's performance at identifying both subtle and bursty anomalies. For subtle anomalies, we use three months of \chicago, from 01/2014 to 03/2014, and inject 50 synthetic bike trips that simulate infrequent, but lasting and surprisingly regular traffic, analogous to the first taxi trip in Fig.~\ref{fig:real-time-anom}. We describe the exact injection procedure in \S~\ref{subsec:anomaly-injection}. Each of the 50 anomalous trips occurs repeatedly, and the task is to identify the occurrences of all 50 bike trips. We average results over ten random injection sets. For bursty anomalies, we use the \darpa network commonly used for the task \cite{eswaran2018sedanspot, bhatia2020midas}. In this dataset, 2.7 million edges correspond to various bursty network attacks (e.g., denial of service). The goal is to identify edges that are part of these attacks. We use the same $\wExp, \fExp$, and $\sExp$ as \S~\ref{subsubsec:realtime-PvF}, since we found them useful for visually identifying bursty and subtly persistent snippets. 

\noindent \textbf{Baselines.} 
{(1)~$\freqBaseline$ scores snippets as their number of occurrences, divided by the total number of all snippet occurrences. It can be thought of as \cite{paranjape2017motifs} with motifs extended to activity snippets. (2)~\sedanspot \cite{eswaran2018sedanspot} and (3)~\midas \cite{bhatia2020midas}, state-of-the-art methods for \emph{bursty} anomalies in edge-streams, output an anomaly score for each edge update. We set parameters as in the respective papers, and use the authors' code. (4) \ds adapts the heuristic persistence of an item in a data-stream \cite{dai2016finding}. 
We apply \ds exactly as \method, but replace the persistence in the 2D point with their heurisitc (described in \S~\ref{subsec:measure-theory}). For consistency, we follow the authors' suggestion of dividing the stream into 60 measurement periods, even though this unrealistically assumes that the stream length is known \emph{a priori}.}

\vspace{0.1cm}
\noindent \textbf{Results.} We give results in Tab.~\ref{table:quant-anomaly-detection}. \methodOnline outperforms baselines at finding subtly persistent anomalies. Simultaneously, for bursty anomalies, 
it performs competitively with \midas, which is designed specifically for the task of finding bursty anomalies. 
On the other hand, the methods targeting bursts do not perform competitively at finding subtle anomalies. 
{All but one result are statistically significant at a  0.01 $p$-value in a paired t-test.}

\subsection{\rThree{}: Efficiency and Scalability}
\label{subsec:scale}
We evaluate whether \methodOnline can process edges at least as quickly as they arrive in a stream, 
while allowing snippets to be reasonably sized ($\maxSize$) and take a reasonable amount of time to form ($\maxDuration$).
We then analyze how \methodOffline scales with the number of edge-updates. We discuss hardware in \S~\ref{subsec:hardware}.

\subsubsection{\methodOnline Efficiency} 
\textbf{Setup.} 
To evaluate \methodOnline's efficiency over different $\maxDuration$ and $\maxSize$ settings, we create plots for $\maxSize \in \{1, 2, 3\}$. In each, we fix $\maxSize$ and vary $\maxDuration \in \{60, 120, 180, $ $300, 600, 900, 1800\}$. We show edges processed per second for each parameter, for all datasets with rates less than 1 update/sec (since streams with significantly different rates are not comparable). Each point is averaged over 5 random intervals of 100K edge-updates (the same intervals across parameters).

\vspace{0.1cm}
\noindent \textbf{Results.} We show the results in Fig.~\ref{fig:edge-updates-per-sec}. For $\maxSize = 1$, activity snippets have duration $\duration = 0$, which is why in the first plot, the edges processed per second is consistent over all $\maxDuration$. Across all streams, and parameters, \methodOnline processes edge-updates 10K to 360K times faster than the rate of the corresponding stream.

\subsubsection{\methodOffline Scalability}

\begin{figure}[t!]
    \centering
    \includegraphics[width=\linewidth]{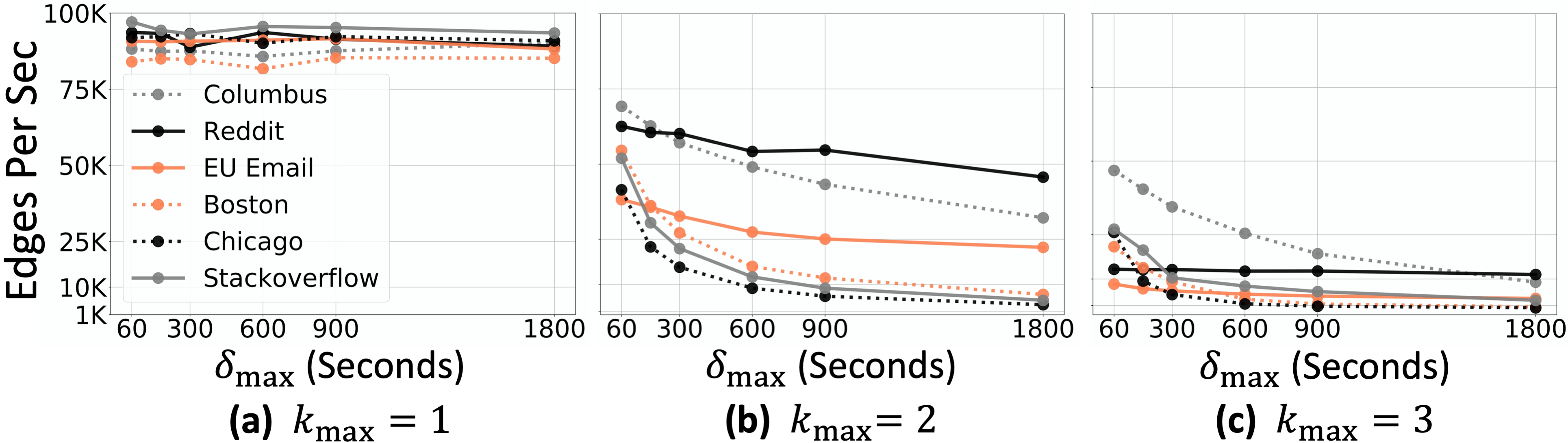}
    \vspace{-0.5cm}
    \caption{{\methodOnline's performance when varying $\maxDuration$, $\maxSize$ on several datasets. 
    Across all parameters, \methodOnline processes edges in each stream 10K-360K times faster than the rate of the stream.}}
    \label{fig:edge-updates-per-sec}
    \vspace{-0.3cm}
\end{figure}

\begin{wrapfigure}{r}{0.45\linewidth}
\vspace{-0.4cm}
    \centering
    \includegraphics[width=\linewidth]{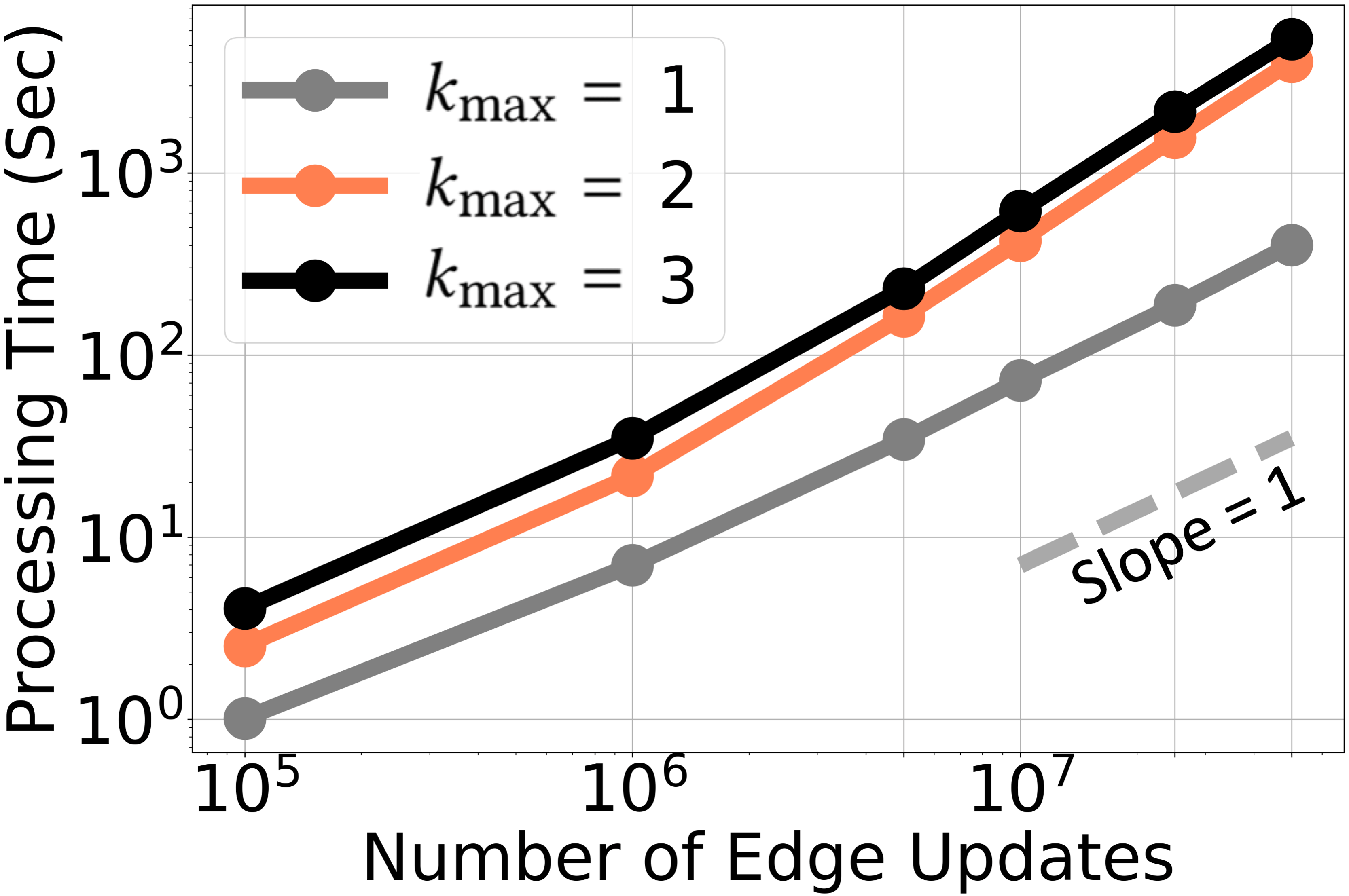}
    \vspace{-0.5cm}
    \caption{\methodOffline scales linearly as the network grows.}
    \label{fig:processing-time}
    \vspace{-0.5cm}
\end{wrapfigure}
\textbf{Setup.} We evaluate how \methodOffline scales with increasingly more edge-updates in \stackoverflow, our network with the most edges and nodes. We process the first 100K, 500K, 1M, 10M, 25M, and 50M edge-updates 5 times, and report the average runtime in seconds (Fig.~\ref{fig:processing-time}). We fix $\maxDuration = 600$ sec (10 min) and evaluate $\maxSize \in \{1, 2, 3\}$.

\vspace{0.1cm}
\noindent \textbf{Result.}  \methodOffline scales linearly as the network grows.

\section{Conclusion}
\label{sec:conclusion}
In this paper, we propose mining persistent activity in continually evolving networks. Our precise, theoretical definition of persistence captures, beyond the aggregate number of occurrences, for how long and how regularly the activity has occurred. 
We propose \method (both offline and streaming variants) to measure the persistence of activity in evolving networks, and use it to 
gain a better understanding of networks by revealing activity that frequency alone could not, from infrequent but surprisingly regular trips in traffic networks to heated conversations in social networks.
{Future work includes further developing persistence-based anomaly detection, and techniques for automatic parameter tuning ($\wExp, \fExp,$ $\sExp$).}

\section*{Acknowledgements}
{
{This work is supported by}
an NSF GRFP 
Fellowship, 
the NSF under Grant No. IIS 1845491, Army Young Investigator Award No. W911NF1810397, 
and Adobe, Amazon, and Google faculty awards.
}
\bibliographystyle{plain}
\bibliography{abbrev,references}

\appendix
\clearpage
\section{Supplement on reproducibility}
\subsection{Complexity Analysis}

In the special case of $\maxSize = 1$, the window $\window$ need not be maintained (only singleton snippets are extracted). Thus, the per-update and total complexity (of \methodOffline) are $\order(1)$ and $\order(|\stream|)$ respectively. We now discuss $\maxSize > 1$. To process a new update $\newUpdate$, the only non-constant cost comes from extracting new snippets in line 4 (lines 5-8 are $\order(1)$). Let $\streamRate$ be the average rate of the stream $\stream$ in updates per second. Then the average number of updates in a window $\window$ of width $\windowWidth = \maxDuration$ seconds is equal to $\streamRate \cdot \maxDuration$. Lines 14-16 are $\order(\streamRate \cdot \maxDuration)$. While new snippets must be connected (cf. \S~\ref{subsubsec:alg-extract-patterns}), in the worst case all previous $\streamRate \cdot \maxDuration$ updates are connected with $\newUpdate$. Thus, there are $\order\left(\sum_{\patternSize = 1}^{\maxSize} { \streamRate \cdot \maxDuration \choose \patternSize - 1 } \right)= \order\left({ \streamRate \cdot \maxDuration \choose \maxSize - 1 }\right)$ new snippets to extract (lines 18-20), which dominates $\order(\streamRate \cdot \maxDuration)$. Consequently, the per-update time can be controlled by choosing $\maxDuration$ to be reasonable based on the stream's rate. For \methodOffline, the total time complexity to process $|\stream|$ updates is $\order(|\stream| { \streamRate \cdot \maxDuration \choose \maxSize - 1 })$.

\subsection{Detailed Pseudocode}
\label{subsec:detailed-pseudocode}

\begin{algorithm}[h]
	\caption{\method($\stream$, $\maxDuration$, $\maxSize$, $\view$, $\wExp$, $\fExp$, $\sExp$, \textsc{variant})}
	\label{alg:method-detailed}
	\begin{algorithmic}[1]
	    \small
	    \Statex \textbf{Input}: Stream $\stream$, max snippet duration $\maxDuration$ and size $\maxSize$, view $\view$, persistence exponents $\wExp$, $\fExp$, $\sExp$, the variant (\methodOffline/\methodOnline).
	    \State $\windowWidth \gets \maxDuration$ \Comment{\textcolor{gray}{The window size enforces the maximum duration}}
	    \State $\window \gets \emptyset$ \Comment{\textcolor{gray}{Window is empty initially}}
	    \While{$\newUpdate \in \stream$} \Comment{\textcolor{gray}{While there is a new update in the stream}}
	        \For{$\pattern \in \Call{ExtractNewSnippets}{\window, \newUpdate, \timestamp, \view}$}
	            \If{\textsc{variant} is \methodOffline}
	                \State $\occs_\pattern \gets \occs_\pattern \cup \{\timestamp\}$ \Comment{\textcolor{gray}{Add the new occurrence's timestamp}}
	            \Else
	                \State Update $\persistence(\pattern;[\streamStart,\timestamp])$ via Thm.~\ref{thm:entropy-incremental}
	            \EndIf
	        \EndFor
	    \EndWhile
	    \If{\textsc{variant} is \methodOffline}
	        \For{each $\pattern$}
	            \State Compute $\persistence(\pattern;[\streamStart, \streamEnd])$
	        \EndFor
	    \EndIf
	\Procedure{ExtractNewSnippets}{$\window, \newUpdate, \timestamp, \view$}
	    \State $\extractedPatterns \gets \{\view(\newUpdate)\}$ \Comment{\textcolor{gray}{Add the singleton snippet for $\newUpdate$}}
	    \For{$\update \in \window$}
	        \If{$\timestamp - \updateTimestamp(\update) > \windowWidth$} 
	            \State $\window \gets \window \setminus \{\update\}$ \Comment{\textcolor{gray}{Remove stale updates}}
	        \EndIf
	    \EndFor
	    \State $\window \gets \window \cup \{\newUpdate\}$ \Comment{\textcolor{gray}{Add the new update}}
	    \For{$\patternSize = {2, \dots, \maxSize}$}
	        \For{each new size $\patternSize$ snippet $\pattern$}
	            \State $\extractedPatterns \gets \extractedPatterns \cup \{\text{\pattern}\}$
	        \EndFor
	    \EndFor
	    \State \Return $\extractedPatterns$
	\EndProcedure
	
	\end{algorithmic}
\end{algorithm}

\subsection{Choosing $\wExp$, $\fExp$, $\sExp$}
\label{subsec:setting-exponents}

This section provides analysis of the exponents $\wExp$, $\fExp$, and $\sExp$. At the end of the section, we give suggestions for practitioners. 

\subsubsection{Rank Correlation Between Persistence and Components}

\begin{figure}[t!]
    \centering
    \includegraphics[width=\linewidth]{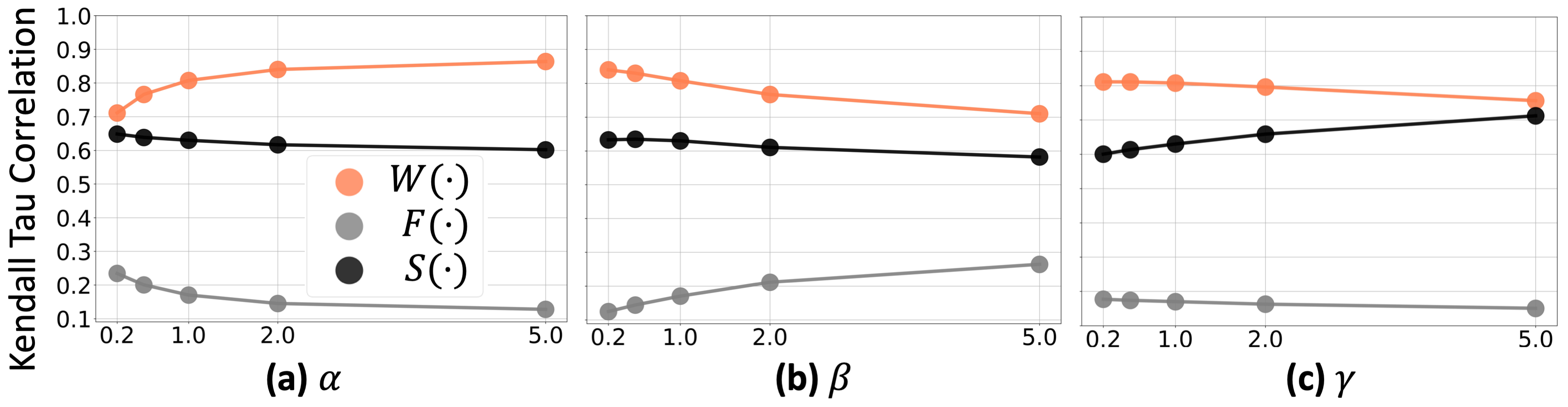}
    \vspace{-0.5cm}
    \caption{Kendall Tau rank correlation between snippets ranked by components of persistence and persistence itself, over various values of exponents. In general, as one exponent is increased, and the others fixed, the corresponding component becomes more correlated with persistence.}
    \label{fig:kendall-tau}
    \vspace{-0.3cm}
\end{figure}

\begin{figure}[t!]
    \vspace{-0.3cm}
	\centering
    \subfloat[$\wExp \in \{0.2, 0.5, 1, 2.0, 5.0\}$\label{fig:params-alpha}]{
      \includegraphics[width=\linewidth]{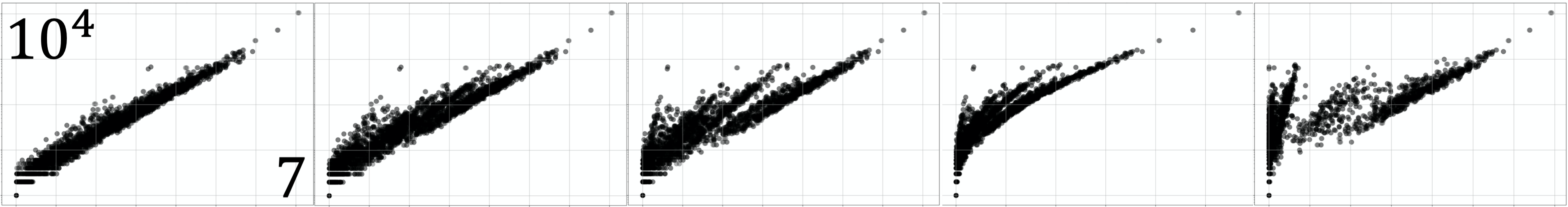}
    }\\
    \vspace{-0.1cm}
    \subfloat[$\fExp \in \{0.2, 0.5, 1, 2.0, 5.0\}$\label{fig:params-beta}]{
      \includegraphics[width=\linewidth]{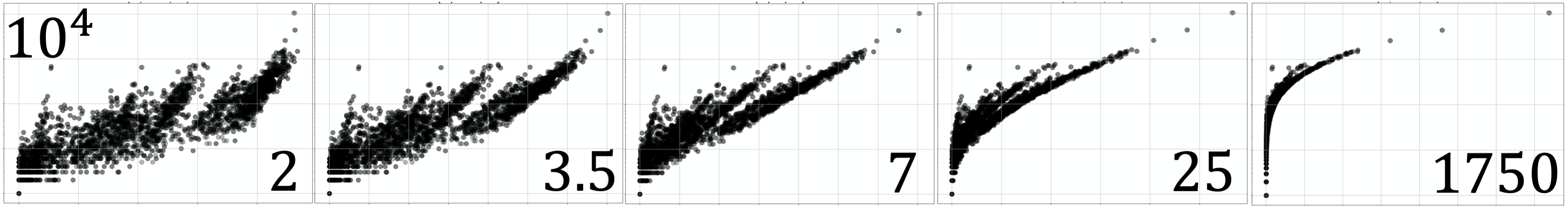}
    }\\
    \vspace{-0.1cm}
    \subfloat[$\sExp \in \{0.2, 0.5, 1, 2.0, 5.0\}$\label{fig:params-gamma}]{
      \includegraphics[width=\linewidth]{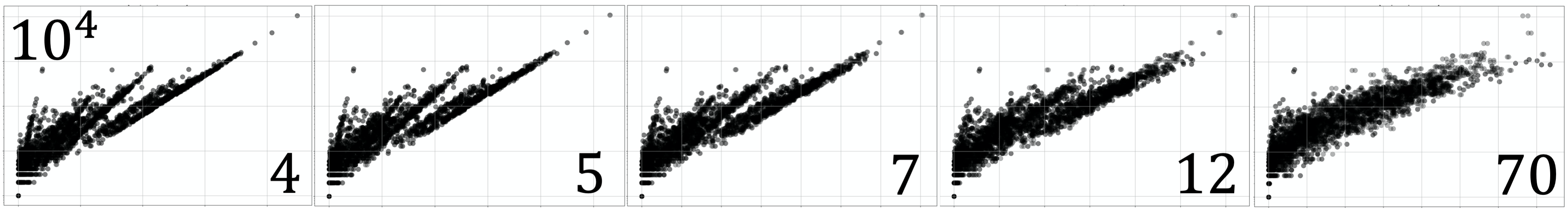}
    }
    \vspace{-0.25cm}
	\caption{\pf plots for \columbus varying each exponent in $\{0.2, 0.5, 1, 2.0, 5.0\}$, while fixing the other two at 1. The main takeaways are that small values of $\fExp$ (0.2 or 0.5) increase the spread of points, while increasing $\sExp$ (2.0 or 5.0) emphasizes points in the lower left (i.e., very regular snippets). Thus, it is generally effective to set $\wExp = 1$, $\fExp \in (0, 1)$, and $\sExp \in (1, \infty)$.}
    \label{fig:params-all}
    \vspace{-0.3cm}
\end{figure}

We first show how varying the exponents affects how much each component contributes to persistence. To do so, we compare the ranking of snippets 
in descending order on $\width(\cdot)$, $\freq(\cdot)$, and $\spread(\cdot)$, with the ranking in descending order based on persistence $\persistence(\cdot)$. We use Kendall-Tau rank-correlation to compare rankings. For each component of persistence $\width(\cdot)$, $\freq(\cdot)$, and $\spread(\cdot)$, we vary its corresponding exponent $\wExp$, $\fExp$, or $\sExp$ over the values $\{0.2, 0.5, 1, 2.0, 5.0\}$, while fixing the other two exponents at 1. Each exponent has a plot in Fig.~\ref{fig:kendall-tau} showing the rank-correlation of that component with persistence. In general, as the exponent corresponding to a component is increased, that component becomes more correlated with persistence, while the others become less correlated.

\subsubsection{Effect on \pf Plots}
\label{subsubsec:visual-parmas}

We next show how varying the exponents changes \pf plots visually. Again, for each component of persistence $\width(\cdot)$, $\freq(\cdot)$, and $\spread(\cdot)$, we vary its corresponding exponent $\wExp$, $\fExp$, or $\sExp$ over $\{0.2, 0.5, 1, 2.0, 5.0\}$, while fixing the others at 1. We show plots for these exponents in Fig.~\ref{fig:params-all}. The value in the upper-left corner is the maximum frequency, and lower-right the maximum persistence. Since $\width(\cdot) \in [0, 1]$, varying $\wExp$ does not change the range of persistence values. In contrast, $\freq(\cdot)$ is unbounded, and increasing it can cause the range of persistence values to grow significantly. The main takeaways are that small values of $\fExp$ (0.2 or 0.5) increase the spread of points, while increasing $\sExp$ (2.0 or 5.0) emphasizes points in the lower left (very regular snippets). We chose exponents to emphasize the snippets of interest in our experiments.

\subsubsection{Sensitivity on Anomaly Detection}

\begin{table}[t]
\centering
	\caption{Additional results at identifying subtly persistent and bursty anomalies, showing the effect of parameters ($\wExp, \fExp, \sExp$).
    }
	\vspace{-0.35cm}
	\label{table:quant-anomaly-detection-full}
	\resizebox{\columnwidth}{!}{
	\begin{tabular}{p{0.5cm}lccccccc}
		\toprule
		& \emph{Metric} & (1, 1, 1) & (0.2, 1, 1) & (1, 0.2, 1) & (1, 1, 0.2) \\
		\toprule
		\multirow{4}{*}{{\rotatebox{90}{\bf Subtle}}}
		& AUC & $0.709 \pm 0.02$ & $0.801 \pm 0.02$ & $0.712 \pm 0.02$ & $0.686 \pm 0.03$ \\
		& F1@100 & $0.006 \pm 0.00$ & $0.009 \pm 0.00$ & $0.005 \pm 0.00$ & $0.007 \pm 0.00$ \\
		& F1@1K & $0.028 \pm 0.00$ & $0.051 \pm 0.01$ & $0.029 \pm 0.00$ & $0.028 \pm 0.01$ \\
		& F1@2K & $0.042 \pm 0.01$ & $0.076 \pm 0.01$ & $0.043 \pm 0.01$ & $0.041 \pm 0.01$ \\
        \midrule
        \multirow{4}{*}{{\rotatebox{90}{\bf Bursty}}}
        & AUC & $0.856 \pm 0.01$ & $0.867 \pm 0.01$ & $0.831 \pm 0.01$ & $0.853 \pm 0.01$ \\
		& F1@500K & $0.307 \pm 0.00$ & $0.307 \pm 0.00$ & $0.307 \pm 0.00$ & $0.307 \pm 0.00$ \\
		& F1@1M & $0.525 \pm 0.00$ & $0.527 \pm 0.00$ & $0.516 \pm 0.01$ & $0.524 \pm 0.00$\\
		& F1@2M & $0.763 \pm 0.01$ & $0.783 \pm 0.01$ & $0.742 \pm 0.02$ & $0.756 \pm 0.01$ \\
  	    \bottomrule
	\end{tabular}
	}
\end{table}

We give anomaly detection results in Tab.~\ref{table:quant-anomaly-detection-full}, showing the effect of downweighting each exponent to $0.2$. The results are mostly stable across exponents. The main exception, $(0.2, 1, 1)$, leads to considerably better results on subtle anomalies. Since $\width(\cdot) \in [0,1]$, $\wExp = 0.2$ \emph{up}-weights $\width(\cdot)$. Since the subtly persistent anomalies (\S~\ref{subsec:anomaly-injection}) occur throughout the stream, we conjecture that $\wExp = 0.2$ increases their anomalousness.

\subsubsection{Advice for Practitioners}
{For visual clarity, we found it generally effective to set $\wExp = 1$, $\fExp \in (0, 1)$, and $\sExp \in (1, \infty)$---cf.~\ref{subsubsec:visual-parmas}. If practitioners wish to analyze \pf plots visually, we recommend these values, especially $\sExp > 2$, to help discover subtly regular snippets. 
If other tasks are of interest, the exponents should be tuned for that task. Indeed, setting $\wExp$, $\fExp$, $\sExp$ automatically for tasks like anomaly detection is an important direction for future work.}

\begin{figure}[t!]
    \centering
    \includegraphics[width=\linewidth]{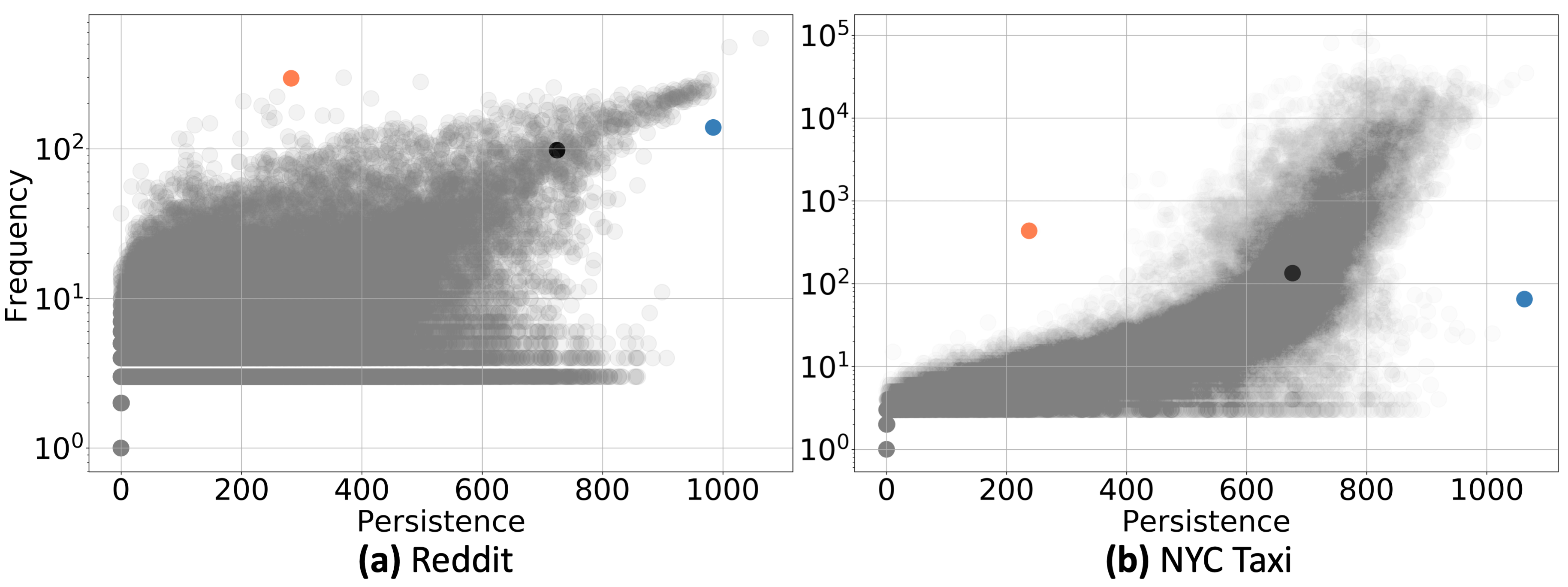}
    \vspace{-0.6cm}
    \caption{The plots used to extract ground-truth anomalies in \S~\ref{subsubsec:realtime-PvF}. Orange is the bursty anomaly, blue the subtly persistent anomaly, black the neither bursty nor subtly persistent snippet.}
    \label{fig:reddit-nyc-supplement}
    \vspace{-0.3cm}
\end{figure}

\subsection{Choosing Activity Snippets}
\label{subsec:pf-supplement}

\begin{wrapfigure}{h!}{0.55\linewidth}
\vspace{-0.4cm}
    \centering
    \includegraphics[width=\linewidth]{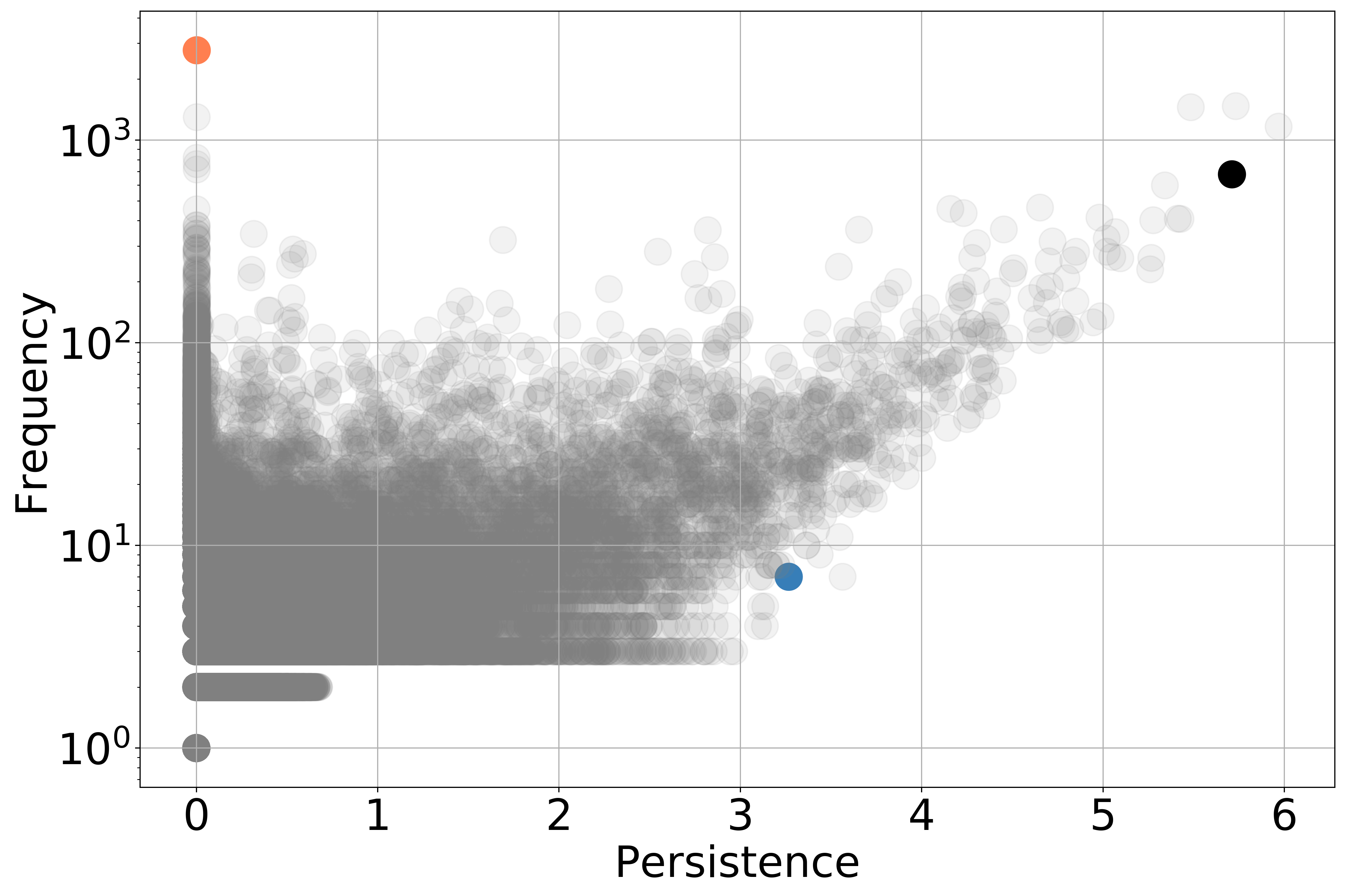}
    \vspace{-0.65cm}
    \caption{The plot used to extract ground-truth for \stackoverflow in \S~\ref{subsubsec:social-networks}. Orange is the bursty anomaly, blue the subtly persistent anomaly, and black a persistent snippet.}
    \label{fig:stackoverflow-supplement}
    \vspace{-0.3cm}
\end{wrapfigure}

We discuss our choices of activity snippets for analysis in \S~\ref{subsubsec:social-networks}-\S~\ref{subsubsec:realtime-PvF}.

\subsubsection{\stackoverflow}
\label{subsubsec:stackoverflow-supplement}

For \stackoverflow, we used the ground-truth points shown in Fig.~\ref{fig:stackoverflow-supplement}. The snippets corresponding to the orange, blue, and black points are visualized in Fig.~\ref{fig:stackoverflow}(a), (i)-(iii).

\subsubsection{\reddit and \NYC}
\label{subsubsec:reddit-supplement}

For \reddit and \NYC, we used the ground-truth snippets in Fig.~\ref{fig:reddit-nyc-supplement}. The orange snippet for \reddit is \odds referencing \hockey, the blue is for \cig referencing \poker, and the black is for \bestof referencing \finance. For NYC, orange is a trip from \taxiCommision to \taxiCommision, blue is a trip from \suspiciousSrc to \suspiciousDst, and black a trip from \nothingSrc to \nothingDst.

\begin{table}[b]
    \vspace{-0.3cm}
    \centering
    \caption{Reference of parameters used in our code for figures and tables reported in \S~\ref{sec:eval}. For maximum size of $\maxSize = 1$, the duration of a snippet is always 0, in which case the maximum duration can be set arbitrarily without affecting results.}
    \label{table:parameters}
    \vspace{-0.35cm}
    \resizebox{\columnwidth}{!}{
        \begin{tabular}{l r r r r r r r}
        \toprule
         Tab./Fig. & $\maxSize$ & $\maxDuration$ & Variant & View  $\view$ & $\wExp$ & $\fExp$ & $\sExp$ \\
        \midrule
        Fig. \ref{fig:bikes} (Boston) & 1 & N/A & \methodOffline & \idView & 1 & 0.5 & 2 \\
        Fig. \ref{fig:bikes} (Columbus) & 1 & N/A & \methodOffline & \idView &  2 & 0.5 & 3 \\
        Fig. \ref{fig:nyc-sandy-f-vs-p} & 1 & N/A & \methodOffline & \idView & 1 & 1 & 10 \\
        Fig. \ref{fig:stackoverflow} (left) & 3 & 3600 & \methodOffline & \idView & 1 & 0.5 & 2 \\
        Fig. \ref{fig:stackoverflow} (right) & 3 & 900 & \methodOffline & \orderView & 2 & 0.5 & 10 \\
        Fig. \ref{fig:real-time-anom} & 1 & N/A & \methodOnline & \idView & 1 & 0.2 & 10\\
        \bottomrule
        \end{tabular}
        }
\end{table}

\subsection{Using Random Cut Forests}
\label{subsec:rcf-supplement}
We discuss details of using Random Cut Forests for anomaly detection in \S~\ref{subsec:realtime}. Throughout the experiments, we use 10 trees in the random forest, with each having a maximum depth of 256. 
To enforce the maximum size of trees, when the maximum size is reached, before adding a new point, we chose a leaf at random to remove. Since activity snippets can reoccur, when scoring a reoccurrence we make one minor adaption. When we score the point $[\freq(\pattern;[\streamStart,\timestamp])$, $\persistence(\pattern;[\streamStart,\timestamp])]$, if we have already scored snippet $\pattern$ at some prior time $\timestamp' < \timestamp$, then we first remove the point $[\freq(\pattern;$$[\streamStart,\timestamp'])$, $\persistence(\pattern;$$[\streamStart,\timestamp'])]$ corresponding to the prior occurrence, to avoid scores decaying artificially due to prior occurrences of the same snippet.

\subsection{Injecting Subtly Persistent Anomalies}
\label{subsec:anomaly-injection}

We use the following procedure to inject subtly persistent anomalies for \S~\ref{subsubsec:quant-anom}. We inject bike trips into the first three months of \chicago. For each anomalous bike trip, we select a start and end position at random within 10 minutes of the start and end of the stream, so that the trip covers most of the three months. We then select a number of occurrences  $|\occs|_\pattern$ from 5 to 100, weighted inversely proportional to the chosen number, to favor lower frequencies. We inject that many anomalies at roughly uniform intervals into the stream, but perturb the gaps from uniform by $\pm 20$ minutes to simulate realistic variance. The anomalous bike trips are chosen from among those not currently present in the stream so that they do not conflict with existing trips. The number of anomalous edges is the sum of the randomly chosen number of occurrences over all 50 bike trips, and these edges are labeled as 1 while the rest are 0. We generated 10 injection sets using different random seeds, and the exact number of resulting edges injected in each set was 3322,
3354,
2714,
2474,
3764,
3366,
3606,
2760,
3360, and
2560.

\subsection{Hardware and Software}
\label{subsec:hardware}

We perform all experiments on an Intel(R) Xeon(R) CPU E5-2697 v3 @ 2.60GHz with 1TB RAM. Our code is implemented in Python.

\subsection{Reference of Parameters Used}
\label{subsec:param-ref}

Table ~\ref{table:parameters} gives the parameters used in each of the experiments.

\end{document}